\documentclass{lmcs}
\pdfoutput=1

\usepackage[utf8]{inputenc}

\usepackage{lastpage}
\lmcsdoi{20}{2}{3}
\lmcsheading{}{\pageref{LastPage}}{}{}%
{Feb.~16,~2022}{Apr.~10,~2024}{}


\usepackage{hyperref}
\usepackage{amsmath}
\usepackage{amssymb}
\usepackage{graphicx}
\usepackage{float}
\usepackage{caption}
\usepackage{amsthm}
\usepackage[linesnumbered,algoruled,boxed,lined]{algorithm2e}

\newcommand{\qq}{\mathbb Q}
\newcommand{\ee}{\mathbb E}
\newcommand{\ff}{\mathbb F}

\newcommand{\abs}[1]{\left|{#1}\right|}

\newcommand{\suchthat}{\ | \ }

\newcommand{\tth}{^\text{th}}
\newcommand{\genseq}[3]{{#1}_1 {#3} {#1}_2 {#3} \dots {#3} {#1}_{#2}}
\newcommand{\seq}[2]{\genseq{#1}{#2}{,}}

\newcommand{\twocases}[4]{\begin{cases} #2 & #1 \\ #4 & #3 \end{cases}}
\newcommand{\threecases}[6]{\begin{cases} #2 & #1 \\ #4 & #3 \\ #6 & #5 \end{cases}}

\newcommand{\txt}[1]{\text{#1}}

\newcommand{\stext}[1]{\ \ \ \ \ \text{(#1)}}
\newcommand{\stextn}[1]{\\&\ \ \ \ \ \ \stext{#1}}
\newcommand{\bcause}[1]{\stext{because ${#1}$}}
\newcommand{\snc}[1]{\stext{since ${#1}$}}

\newcommand{\push}{\\ & \ \ \ \ \ \ \ \ \ \ }

\newcommand{\ipnc}[3]{\begin{figure}[H]\begin{center}\includegraphics[scale = {#1}]{#2.pdf}\caption{#3}\end{center}\end{figure}}

\makeatletter 
\g@addto@macro{\@algocf@init}{\SetKwInOut{Parameter}{Parameters}} 
\makeatother

\newcommand{\PP}{\mathsf{P}}
\newcommand{\NP}{\mathsf{NP}}

\newcommand{\mTwoTwo}[4]{\begin{bmatrix}{#1}&{#2}\\{#3}&{#4}\end{bmatrix}}

\newcommand{\vectTwo}[2]{\begin{bmatrix}{#1}\\{#2}\end{bmatrix}}

\newcommand{\threesat}{\textsf{3SAT}\xspace}
\newcommand{\threexor}{\textsf{3XOR}\xspace}
\newcommand{\vc}{\textsf{VertexCover}\xspace}
\newcommand{\labelcover}{\textsf{LabelCover}\xspace}
\newcommand{\ug}{\textup{\textsf{UniqueGames}}\xspace}

\newcommand{\UG}{\textup{\textsf{UG}}}
\newcommand{\gug}{\textup{\textsf{GroupUniqueGames}}\xspace}

\newcommand{\opt}{\textup{opt}}

\newcommand{\tauuug}{\tau_{\UG(q)}}

\newcommand{\fm}{\ff_2^m}

\newcommand{\soundness}{\frac{1}{2^{2\ell - 1} + 2^{\ell - 1}}}

\begin{document}

\title[Inapproximability of Unique Games in FPC]{Inapproximability of Unique Games in\texorpdfstring{\\}{ }Fixed-Point Logic with Counting}
\titlecomment{{\lsuper*}A conference version of this paper appeared in LICS 2021 under an identical title. ACM subject classification: Theory of computation / Logic / Finite Model Theory}

\author[J.~Tucker-Foltz]{Jamie Tucker-Foltz \lmcsorcid{0000-0001-9174-3341}}
\address{Harvard University}
\email{jtuckerfoltz@gmail.com}

\maketitle

\begin{abstract}
  \noindent We study the extent to which it is possible to approximate the optimal value of a Unique Games instance in Fixed-Point Logic with Counting (FPC). Formally, we prove lower bounds against the accuracy of FPC-interpretations that map Unique Games instances (encoded as relational structures) to rational numbers giving the approximate fraction of constraints that can be satisfied. We prove two new FPC-inexpressibility results for Unique Games: the existence of a $(1/2, 1/3 + \delta)$-inapproximability gap, and inapproximability to within any constant factor. Previous recent work has established similar FPC-inapproximability results for a small handful of other problems. Our construction builds upon some of these ideas, but contains a novel technique.  While most FPC-inexpressibility results are based on variants of the CFI-construction, ours is significantly different. We start with a graph of very large girth and label the edges with random affine vector spaces over $\ff_2$ that determine the constraints in the two structures. Duplicator's strategy involves maintaining a partial isomorphism over a minimal tree that spans the pebbled vertices of the graph.
\end{abstract}

\maketitle

\section{Introduction}\label{secIntro}

While the bulk of ordinary complexity theory rests on conjectures about the limits of efficient computation (e.g., $\PP \neq \NP$), descriptive complexity has had much success in proving \emph{unconditional} lower bounds. By restricting attention only to those algorithms that can be defined in some given primitive logic, it is possible to say much more about what such an algorithm is capable or incapable of doing.

One such logic is \emph{Fixed-Point Logic with Counting (FPC)}, which is first order logic augmented with a least fixed point operator, numeric variables and counting quantifiers asserting that a given number of distinct objects satisfy a given predicate. The set of decision problems definable in FPC forms a proper subset of $\PP$.\footnote{In the presence of a relation interpreted as a total order over the universe, these sets are equal by the Immerman-Vardi theorem \cite{ImmermanVardi1} \cite{ImmermanVardi2}, but this paper is concerned with unordered structures.} Roughly, a polynomial-time algorithm is definable as an FPC-interpretation only if it respects the natural symmetries of its input, without making any arbitrary choices that break those symmetries.\footnote{It is difficult to rigorously define exactly what is meant by ``symmetry breaking." Anderson and Dawar \cite{SymmetricCircuits} give a precise result along these lines, defined in terms of symmetric circuits.} Thus, FPC has become an extremely important and well-studied logic \cite{FPC}, as it seems to elegantly capture the essence of \emph{symmetric computation}. It is a robust logic, including a wide range of powerful algorithmic techniques, such as linear \cite{LinearProgrammingInFPC} and semidefinite \cite{SDPInFPC1} programming.

It was once thought that FPC captured all of polynomial time, until Cai, F\"urer and Immerman \cite{CaiFurerImmerman} constructed a problem that is solvable in polynomial time, yet not definable in FPC. Their proof technique, which has since become known as a \emph{CFI-construction}, was to exhibit pairs of problem instances $(\mathbb{A}_1, \mathbb{B}_1), (\mathbb{A}_2, \mathbb{B}_2), (\mathbb{A}_3, \mathbb{B}_3), \dots$ such that each $\mathbb{A}_k$ is a YES instance, each $\mathbb{B}_k$ is a NO instance, yet $\mathbb{A}_k$ and $\mathbb{B}_k$ are indistinguishable by any sentence of FPC with $k$ variables (formally, they are $C^k$-equivalent: Duplicator wins the $k$-pebble bijective game played on $\mathbb{A}_k$ and $\mathbb{B}_k$; see Section~\ref{subLogics} for the full definition). CFI-constructions have been used to establish logical inexpressibility results for a range of other decision problems---some solvable in polynomial time, others $\NP$-hard \cite{FPC,SolvingEquationsNotInFPC,ThreeColourability}.

In a recent paper, Atserias and Dawar \cite{DefinableInapproximabilityJournal} adapted this technique to obtain the first known \emph{FPC-inapproximability} results, proving the nonexistence of approximation algorithms whose mapping of inputs to outputs is definable as an FPC-interpretation. Essentially, this requires that $\mathbb{A}_k$ and $\mathbb{B}_k$ not only have different optimal values, but that these optimal values differ by some constant multiplicative factor. Atserias and Dawar were able to construct such instances to derive a tight FPC-inapproximability gap for \threexor (like \threesat but with XOR's instead of OR's in each clause), which was then extended via first order interpretations to yield lower bounds for \threesat, \vc and \labelcover.

\subsection{Main results}

In this paper, we extend this line of work to consider a keystone problem in the theory of approximation algorithms: \ug. A \ug instance $U$ is specified by a set of variables taking values in some fixed label set $[q] = \{1, 2, \dots, q\}$, and a set of constraints. Each constraint requires that a certain pair of variables take values that are consistent with some permutation on $[q]$. The goal is to find the maximum fraction of constraints that can be simultaneously satisfied.\footnote{The \ug problem is so named because, given any value for one variable in a constraint, there is a \textbf{unique} value for the other variable that satisfies the constraint (defined by the permutation on the label set associated to the constraint). There is an alternative interpretation of the optimization problem in terms of finding an optimal strategy in a \textbf{game} between two provers and a verifier, where the provers wish to convince the verifier that there is an assignment satisfying all constraints. See the survey by Khot \cite{UGCSurvey} for more background on \ug.} We denote this fraction by $\opt(U)$. Our main result (Theorem~\ref{thmMain}) is that, for any $\delta > 0$ and any positive integer $\ell$, for sufficiently large $q$ we can construct pairs of \ug instances $(\mathbb{A}_1, \mathbb{B}_1), (\mathbb{A}_2, \mathbb{B}_2), (\mathbb{A}_3, \mathbb{B}_3), \dots$ on a label set of size $q$ such that, for all $k$, $\mathbb{A}_k$ and $\mathbb{B}_k$ are $C^k$-equivalent, yet
\begin{align*}
\opt(\mathbb{A}_k) \geq \frac{1}{2^\ell}, && \opt(\mathbb{B}_k) < \soundness + \delta.
\end{align*}

There are two important specializations of the parameter $\ell$. Sending $\ell \to \infty$, it follows that there is no constant-factor approximation algorithm for \ug that is definable as an FPC-interpretation (Corollary~\ref{corUGLowGapMain}). As discussed, this does not rule out the existence of \emph{any} polynomial-time constant-factor approximation algorithm, but it does rule out the existence of a \emph{symmetric} algorithm, including any algorithm based primarily on linear or semidefinite programming (i.e., an algorithm that just returns the optimal value of a linear or semidefinite program constructed from the input via elementary operations).

Setting $\ell = 1$, this result implies that no sentence of FPC can distinguish all instances of optimal value $\frac12$ from those of optimal value $\frac13 + \delta$. The analogue of such a result in ordinary complexity theory would be that \ug has a \emph{$(\frac12, \frac13 + \delta)$-inapproximability gap}, i.e., assuming $\PP \neq \NP$, there is no polynomial-time algorithm distinguishing instances of optimal value $\frac12$ from instances of optimal value $\frac13 + \delta$. Our result is, of course, logically incomparable to results of this kind since we have imposed an FPC-definability requirement and dropped the assumption that $\PP \neq \NP$. If we were to compare them modulo this exchange of assumptions, our result would be stronger than what was known prior to the resolution of the 2-2 Games Conjecture in 2018.\footnote{The 2-2 Games Theorem \cite{UGMidGap2018-2} \cite{UGMidGap2018} implies a $(\frac12, \delta)$-inapproximability gap for \ug, while the previous best was $(\frac12, \frac38 + \delta)$ from 2012, due to O'Donnell and Wright \cite{UGMidGap2012}.} However, while the 2-2 Games Theorem is the culmination of over a decade of research spanning several papers, our FPC lower bound is self-contained and qualitatively very different, exploiting a peculiar weakness of FPC: the inability to perform Gaussian elimination. The result is thus a testament to the power of descriptive complexity in proving lower bounds for important restricted classes of algorithms, even in the realm of approximation.

\subsection{Connections to semidefinite programming}

The prominence of \ug as a problem of study in complexity theory is primarily due to its deep connections with semidefinite programming. A famous theorem of Raghavendra \cite{RagThesis} establishes that, for any constraint satisfaction problem $\Lambda$, if we assume $\PP \neq \NP$ and the \emph{Unique Games Conjecture (UGC)}, the best polynomial-time approximation algorithm for $\Lambda$ is given by solving and rounding a specific semidefinite programming relaxation. Thus, proving the UGC would immediately yield a complete classification of the approximability of constraint satisfaction problems, closing the gaps between the best known upper and lower bounds on the optimal approximation ratios (though still conditional on $\PP \neq \NP$). Raghavendra's theorem is established via a general, abstract reduction from \ug to $\Lambda$. As it so happens, this reduction is definable as an FPC-interpretation. Together with the fact that the optimal value of a semidefinite program can be defined in FPC \cite{SDPInFPC1}, this implies an analogue of Raghavendra's result for algorithms that are definable in FPC, holding without the assumption that $\PP \neq \NP$ \cite{Thesis}.

Thus, if we could prove an FPC-version of the UGC, we would have a complete and unconditional understanding of the limits of symmetric computation in approximating constraint satisfaction problems. This conjecture is as follows.

\begin{conj}[FPC-UGC \cite{Thesis}]\label{cnjFPCUGC}
 	For all $\varepsilon, \delta > 0$, there exists $q$ such that there is no sentence $\phi$ of FPC such that, for any \ug instance $\mathbb{A}$ on a label set of size $q$,
 	\begin{enumerate}
 		\item\label{itmFPCUGCCompleteness} if $\opt(\mathbb{A}) \geq 1 - \varepsilon$, then $\mathbb{A} \models \phi$, and
 		\item\label{itmFPCUGCSoundness} if $\opt(\mathbb{A}) < \delta$, then $\mathbb{A} \not\models \phi$.
 	\end{enumerate}
\end{conj}

If true, this would also imply that no fixed level of the Lasserre hierarchy can disprove the ordinary UGC, which is a well-known open problem in complexity theory.\footnote{This problem is described in \cite{LassereDisproveUGCQuestion}. See \cite{SDPInFPC1} for a discussion of the relationship between FPC and the Lasserre hierarchy.}

The main result of this paper is a step in the direction of proving this conjecture. Instead of a $(1 - \varepsilon, \delta)$-gap, we have a $(\frac{1}{2^\ell}, \soundness + \delta)$-gap for any positive integer $\ell$. In light of our aforementioned comparison to existing conditional lower bounds on \ug in the realm of polynomial-time computation, we believe there is promise that the FPC-UGC may be resolved before the ordinary UGC.

\subsection{Challenges and techniques}\label{subTechniques}

In typical CFI-constructions, $\mathbb{A}_k$ and $\mathbb{B}_k$ are exactly the same except at one small place. The difference is critical in that it causes $\mathbb{A}_k$ and $\mathbb{B}_k$ to differ with respect to the property being shown to be inexpressible, yet it is small enough that Duplicator can repeatedly move it around to somewhere else, hiding it from Spoiler. For \ug (and probably any other optimization problem of practical interest), two structures differing only on a small number of constraints must have nearly the same optimal value. Since we require instances whose optimal values differ significantly, this standard approach does not work for our setting. Duplicator will never be able to maintain an invariant that the two structures look isomorphic everywhere except in one place.

Atserias and Dawar \cite{DefinableInapproximabilityJournal} overcame this difficulty for \threexor by constructing a 3CNF Boolean formula from the \threexor instance and applying theorems on resolution proofs and their connection to existential pebble games. This approach does not work for \ug, since the variables are not Boolean, and the constraints are only between pairs of variables, not triples. Thus, our construction is quite different from existing CFI-constructions in the literature.

The main idea is to start from a graph of large girth and label each edge with a random $\ell$-dimensional subspace of $\fm$ (the $m$-dimensional vector space over the 2-element field) for some large integer $m > \ell$. These subspaces determine the constraints in the more-satisfiable instance $\mathbb{A}_k$. To obtain the less-satisfiable instance $\mathbb{B}_k$, we transform these subspaces into affine subspaces, shifted by random vectors. Duplicator's strategy in the $k$-pebble bijective game relies on the fact that, with high probability, along \emph{any} path $p$ that is sufficiently long (though less than the girth), the union of all of the subspaces along $p$ spans $\fm$. This can be used to ``fill in" the bijection between $\mathbb{A}_k$ and $\mathbb{B}_k$ along long paths, so all Duplicator has to worry about are short paths between pebbled vertices. Duplicator can ensure consistency across these short paths by maintaining an appropriate invariant throughout the game.

\subsection{Organization}

In Section~\ref{secPrelim} we review the necessary background on \ug, FPC, and approximation algorithms. In Section~\ref{secLiftedGraph} we adapt a key tool from the literature on CFI-constructions for satisfiability of systems of linear equations to our setting. Section~\ref{secExamples} works through two examples illustrating some of the main ideas of our construction. The formal proof of $C^k$-equivalence is presented in Section~\ref{secMainProof}, and conclusions follow in Section~\ref{secConclusion}.

\section{Preliminaries}\label{secPrelim}

We assume the reader is familiar with relational structures and first-order (FO) logic. All graphs are undirected, but may contain multiple edges between a pair of vertices and/or loops from a vertex to itself. A graph is \emph{simple} if it has no multiple edges or self-loops. Paths are not allowed to repeat edges.

\subsection{Unique games}\label{subUGBackground}

We denote the maximization problem of \ug on a label set of size $q$ by \UG($q$). We view the label set as $[q] := \{1, 2, \dots, q\}$. We encode \UG($q$) instances as relational structures using a vocabulary $\tauuug$ consisting of $q!$ relation symbols, each of arity 2: for every permutation $\pi: [q] \to [q]$, $\tauuug$ has a relation symbol $P_\pi$. The universe can be thought of as the set of variables of the CSP, and each $P_\pi$ is interpreted as the set of pairs of variables $(x_1, x_2)$ to which we have a constraint that $x_1$ and $x_2$ should take values $z_1$ and $z_2$ in $[q]$ such that $\pi(z_1) = z_2$. These are precisely the kinds of constraints that are allowed in \ug. For any \UG($q$) instance $U$, the objective is to compute $\opt(U) \in [0, 1]$, the maximum fraction of constraints that can be simultaneously satisfied by any assignment of variables to values in $[q]$. We call $\opt(U)$ the \emph{optimal value}, or \emph{satisfiability}, of $U$.

\begin{figure}\begin{center}
		\includegraphics[scale=1.3]{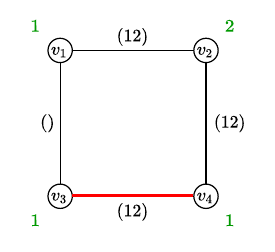}
		\caption{\label{figUGExample} A \ug instance over the label set $\{1, 2\}$ represented graphically, along with one optimal solution (green). Only the bottom edge (red) is unsatisfied by this solution.}
\end{center}\end{figure}

It is often convenient to think of \UG($q$) instances as being defined over some graph $H$, where each vertex $v \in V(H)$ represents a variable $x_v$, and each edge $\{v_1, v_2\} \in E(H)$ is associated to a permutation $\pi_{v_1, v_2}$ on the label set $[q]$ defining which labels for variable $x_{v_1}$ correspond to which labels for variable $x_{v_2}$. When the permutations are involutions (as will be the case throughout this paper), the orientation of the edges does not matter, so we can associate each permutation to an edge, rather than an ordered pair of vertices. For example, Figure~\ref{figUGExample} shows a \UG(2) instance $U$ with edge permutations written in cycle notation, along with one of the optimal assignments ($x_{v_1} = x_{v_3} = x_{v_4} = 1$; $x_{v_2} = 2$) realizing $\opt(U) = \frac34$. If we were to encode this instance as a $\tau_{\UG(2)}$-structure
\[\mathbb{A} = \langle A, P_{()}^\mathbb{A}, P_{(12)}^\mathbb{A} \rangle,\]
the universe would be $A = \{x_{v_1}, x_{v_2}, x_{v_3}, x_{v_4}\}$, and the relations would be interpreted as
\begin{align*}
P_{()}^\mathbb{A} &= \{(x_{v_1}, x_{v_3}), (x_{v_3}, x_{v_1})\},\\
P_{(12)}^\mathbb{A} &= \{(x_{v_1}, x_{v_2}), (x_{v_2}, x_{v_1}), (x_{v_2}, x_{v_4}), (x_{v_4}, x_{v_2}), (x_{v_3}, x_{v_4}), (x_{v_4}, x_{v_3})\}.
\end{align*}
Note in this case the relations happen to be disjoint, which need not be true in general. Indeed, most of the instances constructed in this paper involve multiple, contradictory constraints between some pairs of variables.

A special subclass of \ug instances play a particularly important role, both in this paper and in other papers on \ug \cite{UGCSurvey}. We call these \gug instances, which satisfy the following additional properties:
\begin{enumerate}
	\item\label{itmGUGVertexLabels} The label set $[q]$ is identified with some finite Abelian group $A$ of order $q$.
	\item\label{itmGUGEdgeLabels} For every permutation constraint $\pi$, there is some $g \in A$ such that $\pi(x) = g + x$ (we always write the group operation additively). Thus, we can identify the set of permutations with $A$ as well.
\end{enumerate}

It is convenient to express constraints as equations in $A$. For instance, if $\pi$ is the permutation from condition (\ref{itmGUGEdgeLabels}) above and $(x_{v_1}, x_{v_2}) \in P_\pi^\mathbb{A}$ for some $\tauuug$-structure $\mathbb{A}$, then we would write this constraint as the equation
\[x_{v_2} = g + x_{v_1}.\]
We depict such a constraint graphically by writing $g$ on the (oriented) edge $\{v_1, v_2\}$ (see Figures~\ref{figKEquals2Example} and~\ref{figK4Example} for example). Henceforth we will informally speak of \gug instances $U$ and equations that occur as constraints of $U$, with the implicit understanding of how these are formally represented as $\tauuug$-structures $\mathbb{A}$ and relations $P_\pi^\mathbb{A}$.

\subsection{The logics FPC and $C^k$}\label{subLogics}

Since we are only concerned with lower bounds, it is not necessary to go into the formal definitions of any of the logics discussed in this paper. The main important fact we require is that, for any FPC sentence $\phi$, there is a number $k$ such that $\phi$ can be translated into $C^k$, the fragment of infinitary FO logic with counting quantifiers consisting of (possibly infinite) sentences with only $k$ variables \cite{InfinitaryTranslation}. Therefore, to show that a given property $\mathcal{P}$ is not expressible in FPC, it suffices to construct, for any $k$, a pair of structures $\mathbb{A} = \mathbb{A}_k$ and $\mathbb{B} = \mathbb{B}_k$ such that $\mathbb{A}$ has property $\mathcal{P}$ but $\mathbb{B}$ does not, yet no sentence of $C^k$ can distinguish $\mathbb{A}$ from $\mathbb{B}$, in the sense that $\mathbb{A}$ models any $C^k$ sentence if and only if $\mathbb{B}$ does. When this is the case, we say that $\mathbb{A}$ and $\mathbb{B}$ are \emph{$C^k$-equivalent}, written $\mathbb{A} \equiv_{C^k} \mathbb{B}$.

There is a useful characterization of $C^k$-equivalence in terms of a game between two players, Spoiler and Duplicator, called the \emph{$k$-pebble bijective game}. The board on which they play consists of the universe $A$ of structure $\mathbb{A}$ and  the universe $B$ of structure $\mathbb{B}$. There are $k$ pairs of \emph{pebbles}, initially not placed anywhere. Throughout the game, the pairs of pebbles will be placed on elements of the two universes, one pebble in each universe. Each round of the game consists of three parts:
\begin{enumerate}
	\item\label{itmPebbleGame1} Spoiler picks up one of the $k$ pairs of pebbles, removing them from the board.
	\item\label{itmPebbleGame2} Duplicator gives a bijection $f: A \to B$ such that, for all $1 \leq i \leq k$, if the $i\tth$ pebble pair is placed on some pair of elements $a_i \in A$, $b_i \in B$, then $f(a_i) = b_i$.
	\item\label{itmPebbleGame3} Spoiler places the pebbles back down, placing one pebble on some $a \in A$ and the other pebble on $f(a) \in B$.
\end{enumerate}
At the end of a round, Spoiler wins if the map sending each pebbled element in $A$ to its correspondingly-pebbled element in $B$ is not a \emph{partial isomorphism} between the two structures, i.e., there is some relation in one of the two structures that holds of a set of pebbled elements, but the corresponding relation does not hold in the other structure of the correspondingly-pebbled elements. If Spoiler is unable to win the game in any finite number of moves, then Duplicator wins.

\begin{thmC}[\cite{BijectiveGame}]\label{thmHella}
	Duplicator has a winning strategy in the $k$-pebble bijective game played on $\mathbb{A}$ and $\mathbb{B}$ if and only if $\mathbb{A} \equiv_{C^k} \mathbb{B}$.
\end{thmC}

\subsection{Definability of algorithms}

In descriptive complexity, the logical analogue of a decision algorithm is a sentence. To capture algorithms with non-boolean outputs, we turn to \emph{interpretations}. Suppose we fix finite, relational vocabularies $\sigma$ for the input type and $\tau$ for the output type. By saying that an algorithm is \emph{FPC-definable}, we mean that the mapping from $\sigma$-structures to $\tau$-structures defined by the algorithm can be realized as an \emph{FPC-interpretation of $\tau$ in $\sigma$}.

Roughly, an FPC-interpretation of $\tau$ in $\sigma$ is a formal way to construct a $\tau$-structure $\mathbb{B}$ from a $\sigma$-structure $\mathbb{A}$ by defining the universe of $\mathbb{B}$ in terms of copies of the universe of $\mathbb{A}$, with FPC $\sigma$-formulas defining membership for each relation in $\tau$. For the formal details, see, for example, Dawar and Wang \cite{SDPInFPC1}. Theorem 6 of \cite{SDPInFPC1} contains an example of an FPC-definable algorithm.

\subsection{Inapproximability}\label{subInapproximability}

Let $\Lambda$ denote an arbitrary maximization problem (e.g., think of $\Lambda$ as \UG($q$) for some $q$). For $\alpha \in [0, 1]$, an \emph{$\alpha$-approximation algorithm} for $\Lambda$ is a polynomial-time algorithm that, given an instance of $\Lambda$ with optimal value $x^*$, returns a value $x$ such that $\alpha x^* \leq x \leq x^*$. Typically, an $\alpha$-approximation algorithm works by finding a valid ``solution" to the instance of $\Lambda$ (for \ug, this would be an assignment of values to variables) and returning the value of that solution. The solution may not be optimal, but it should be guaranteed to have value at least an $\alpha$-fraction of that of the optimal solution, whatever that value may be. The closer $\alpha$ is to 1, the better the guarantee of the algorithm.

Let us fix vocabularies $\tau_{\Lambda}$ for instances of $\Lambda$ and $\tau_\qq$ for rational numbers, so that an FPC-definable approximation algorithm for $\Lambda$ is an FPC-interpretation of $\tau_\qq$ in $\tau_\Lambda$. The following lemma is a logical formulation of the well-known\footnote{See the introduction by Vazirani \cite{VaziraniBook} for example.} relationship between approximation algorithms and gap problems.

\begin{lem}\label{lemGapToInapproximability}
	Suppose, for some rational numbers $c \geq s \geq 0$, no sentence of FPC can distinguish $\tau_\Lambda$-structures with optimal value $\geq c$ from those with optimal value $< s$. Then there is no FPC-definable $\alpha$-approximation algorithm for $\Lambda$ for any $\alpha \geq \frac{s}{c}$.
\end{lem}

\begin{proof}
	We prove the contrapositive. Suppose there was such an interpretation $\Theta$ of $\tau_\qq$ in $\tau_\Lambda$. Then we can use $\Theta$ to write an FPC sentence $\phi$ in vocabulary $\tau_\Lambda$ expressing the property that, for a given $\tau_\Lambda$ structure $\mathbb{A}$, $\Theta(\mathbb{A}) \geq s$. (By the Immerman-Vardi theorem, inequality comparison of rational numbers can be done in fixed-point logic even without needing counting quantifiers, since numbers are encoded in ordered structures.) Let $\mathbb{A}$ be any $\tau_\Lambda$ structure with optimal value $x^*$ such that either
	\begin{enumerate}
		\item\label{itmGapProblemCompleteness} $x^* \geq c$, or
		\item\label{itmGapProblemSoundness} $x^* < s$.
	\end{enumerate}
	Let $x := \Theta(\mathbb{A})$. If $\mathbb{A} \models \phi$, we have $x^* \geq x \geq s$. Therefore, we cannot be in case (\ref{itmGapProblemSoundness}), so we must be in case (\ref{itmGapProblemCompleteness}). Conversely, if $\mathbb{A} \not\models \phi$, we have
	\[\alpha x^* \leq x < s \leq c\alpha.\]
	Therefore, $x^* < c$, so we cannot possibly be in case (\ref{itmGapProblemCompleteness}), and hence must be in case (\ref{itmGapProblemSoundness}). Thus, $\phi$ distinguishes the two cases.
\end{proof}

So, just as in ordinary complexity theory, where hardness of a gap problem (that is, distinguishing the two cases in Lemma~\ref{lemGapToInapproximability}), implies hardness of approximation, FPC-inexpressibility of a gap problem implies the nonexistence of FPC-definable approximation algorithms. In either case, the goal in proving lower bounds on algorithms/logics is always to establish as wide a gap as possible.

\section{The label-lifted instance}\label{secLiftedGraph}

Let $U$ be a \gug instance with group $A$ and variable set
\[\{x_v \suchthat v \in V\}.\]
Then we define $\mathcal{G}(U)$ to be a \gug instance with group $A$ and variable set
\[\{x_v^g \suchthat v \in V,\ g \in A\}.\]
For every equation
\[x_{v_1} - x_{v_2} = z\]
in the constraint set of $U$ and every $g_1, g_2 \in A$, we have the equation
\[(x_{v_1}^{g_1} - g_1) - (x_{v_2}^{g_2} - g_2) = z\]
in the constraint set of $\mathcal{G}(U)$. We call $\mathcal{G}(U)$ the \emph{label-lifted instance\footnote{This construction is similar to the \emph{label-extended graph} of a \ug instance (see, for example, \cite{LabelExtendedGraphExample, UGAT}), but it is not the same thing. The label-extended graph is obtained by taking all of the edges with identity constraints in the label-lifted instance. } of $U$}.

The operator $\mathcal{G}$ is similar to the operator $G$ used by Atserias and Dawar \cite[Sec.\txt{} 3.2]{DefinableInapproximabilityJournal}, and also implicitly used by Atserias, Bulatov and Dawar \cite[Sec.\txt{} 3]{SolvingEquationsNotInFPC}. The point is that the extra structure gives Duplicator a new possible strategy to win the $k$-pebble bijective game played on $\mathcal{G}(U_1)$ and $\mathcal{G}(U_2)$ which may not be possible with the original instances $U_1$ and $U_2$; while at the same time, applying $\mathcal{G}$ does not change the satisfiability.

\begin{lem}\label{lemGSameSatisfiability}
	For any \gug instance $U$, $\opt(\mathcal{G}(U)) = \opt(U)$.
\end{lem}

\begin{proof}
	We begin by introducing some notation which is not used outside of this proof. Suppose $U$ has $n$ variables, denoted $x_{v_1}, x_{v_2}, \dots, x_{v_n}$, $C$ constraints, and $\abs{A} = q$. For all $i, j \in [n]$, write $c(i, j)$ for the number of constraints between variables $x_{v_i}$ and $x_{v_j}$, and enumerate them as
	\[\{x_{v_i} - x_{v_j} = z_{i, j, k} \suchthat k \in [c(i, j)]\}.\]
	In the context of some fixed assignment of variables, for any constraint equation $\beta$, let $\mathbb{I}(\beta)$ be the function that evaluates to 1 if $\beta$ is satisfied under the assignment and 0 if $\beta$ is not satisfied.
	
	Let $x_v$ be an assignment\footnote{We sometimes use a symbol like $x_v$ or $x_v^g$ to denote a specific variable, and sometimes to denote the value assigned to that variable. When $v$ is unspecified, as it is here, we mean a function assigning a value to each variable. It should be clear from context which of the three meanings we intend. } of variables attaining the optimum satisfiability of $U$, i.e.,
	\[\opt(U) = \frac1C \sum_{i, j \in [n]}\ \sum_{k \in [c(i, j)]} \mathbb{I}(x_{v_i} - x_{v_j} = z_{i, j, k}).\]
	From this, define an assignment of variables of $\mathcal{G}(U)$ by
	\[x_{v}^{g} := x_{v} + g.\]
	Then the optimal value of $\mathcal{G}(U)$ is at least the number of constraints satisfied by this assignment divided by the total number of constraints $\frac{1}{Cq^2}$, i.e.,
	{\allowdisplaybreaks\begin{align*}
	\opt(\mathcal{G}(U)) &\geq \frac{1}{Cq^2}\sum_{i, j \in [n]}\ \sum_{g_i, g_j \in A}\ \sum_{k \in [c(i, j)]} \mathbb{I}((x_{v_i}^{g_i} - g_i) - (x_{v_j}^{g_j} - g_j) = z_{i, j, k})\\
	&= \frac{1}{Cq^2}\sum_{i, j \in [n]}\ \sum_{g_i, g_j \in A}\ \sum_{k \in [c(i, j)]} \mathbb{I}(((x_{v_i} + g_i) - g_i) - ((x_{v_j} + g_j) - g_j) = z_{i, j, k})\\
	&= \frac{1}{Cq^2}\sum_{i, j \in [n]}\ \sum_{g_i, g_j \in A}\ \sum_{k \in [c(i, j)]} \mathbb{I}(x_{v_i} - x_{v_j} = z_{i, j, k})\\
	&= \frac{1}{Cq^2}\sum_{i, j \in [n]} (q^2) \sum_{k \in [c(i, j)]} \mathbb{I}(x_{v_i} - x_{v_j} = z_{i, j, k})\\
	&= \frac{1}{C}\sum_{i, j \in [n]} \sum_{k \in [c(i, j)]} \mathbb{I}(x_{v_i} - x_{v_j} = z_{i, j, k})\\
	&= \opt(U).
	\end{align*}}
	
	For the other direction, let $x_v^g$ be an assignment of variables attaining the optimum satisfiability of $\mathcal{G}(U)$. Then
	\begin{align*}
	\opt(\mathcal{G}(U)) &= \frac{1}{Cq^2}\sum_{i, j \in [n]}\ \sum_{g_i, g_j \in A}\ \sum_{k \in [c(i, j)]} \mathbb{I}((x_{v_i}^{g_i} - g_i) - (x_{v_j}^{g_j} - g_j) = z_{i, j, k})\\
	&= \frac{1}{Cq^{n}} \sum_{i, j \in [n]}\ \sum_{\seq{g}{n} \in A}\ \sum_{k \in [c(i, j)]} \mathbb{I}((x_{v_i}^{g_i} - g_i) - (x_{v_j}^{g_j} - g_j) = z_{i, j, k}),
	\end{align*}
	since, for every fixed $i, j \in [n]$, the term
	\[\sum_{k \in [c(i, j)]} \mathbb{I}((x_{v_i}^{g_i} - g_i) - (x_{v_j}^{g_j} - g_j) = z_{i, j, k})\]
	is counted exactly $q^{n - 2}$ times. Rearranging the order of summation, we have
	\begin{equation}\label{equGSameSatisfiability1}
	\opt(\mathcal{G}(U)) = \frac{1}{Cq^n} \sum_{\seq{g}{n} \in A}\ \sum_{i, j \in [n]}\ \sum_{k \in [c(i, j)]} \mathbb{I}((x_{v_i}^{g_i} - g_i) - (x_{v_j}^{g_j} - g_j) = z_{i, j, k})
	\end{equation}
	By the averaging principle, there must be some fixed $\seq{g}{n} \in A$ such that
	\begin{equation}\label{equGSameSatisfiability2}
	\sum_{i, j \in [n]}\ \sum_{k \in [c(i, j)]} \mathbb{I}((x_{v_i}^{g_i} - g_i) - (x_{v_j}^{g_j} - g_j) = z_{i, j, k}) \geq C \cdot \opt(\mathcal{G}(U)),
	\end{equation}
	for otherwise, if all of the $q^n$ choices of $\seq{g}{n} \in A$ failed to satisfy (\ref{equGSameSatisfiability2}), we could strictly upper-bound the right-hand side of (\ref{equGSameSatisfiability1}) by
	\[\frac{1}{Cq^{n}}(q^n)\left(C \cdot \opt(\mathcal{G}(U))\right) = \opt(\mathcal{G}(U)),\]
	contradicting (\ref{equGSameSatisfiability1}). Using these fixed $g_i$ values, we define an assignment of variables of $U$ by
	\[x_{v_i} := x_{v_i}^{g_i} - g_i.\]
	It then follows that the optimal value of $U$ is at least the fraction of constraints satisfied by this assignment, i.e.,
	\begin{align*}
	\opt(U) &\geq \frac1C\sum_{i, j \in [n]}\ \sum_{k \in [c(i, j)]} \mathbb{I}(x_{v_i} - x_{v_j} = z_{i, j, k})\\
	&= \frac1C \sum_{i, j \in [n]}\ \sum_{k \in [c(i, j)]} \mathbb{I}((x_{v_i}^{g_i} - g_i) - (x_{v_j}^{g_j} - g_j) = z_{i, j, k})\\
	&\geq \opt(\mathcal{G}(U)) \stext{by (\ref{equGSameSatisfiability2})},
	\end{align*}
	as desired.\footnote{We remark that, with very slight modification, this argument also shows that the $G$ operator of Atserias and Dawar \cite{DefinableInapproximabilityJournal} preserves the exact satisfiability of a \threexor instance. In other words, part (2) of Lemma 3 of \cite{DefinableInapproximabilityJournal} can be strengthened, and as a consequence, the third paragraph in the proof of Lemma 4 of \cite{DefinableInapproximabilityJournal} is unnecessary.}
\end{proof}

\section{Examples}\label{secExamples}

Before presenting the general gap construction in the next section, we first describe some special cases: pairs of $C^k$-equivalent unique games instances with different optimal values, for $k \in \{2, 3\}$.

\subsection{An example where $k = 2$}\label{subKEquals2}

The label-lifted instance is useful for Duplicator since it makes any two \gug instances on the same underlying graph (in the sense of Section~\ref{subUGBackground}) locally isomorphic. For example, take $U_1$ and $U_2$ to be as in Figure~\ref{figKEquals2Example}.

For both $U_1$ and $U_2$, the underlying graph $H$ is the complete graph on 3 vertices. The group is $A := \ff_2$ (the finite field of two elements, $\{0, 1\}$, under addition modulo 2), so there are only two kinds of constraints. The identity constraints are drawn in green and the non-identity constraints are drawn in red. Instance $U_1$ has only identity constraints, so is completely satisfiable by setting all variables to be the same value. In $U_2$, the 3 constraints are inconsistent, but any 2 of them can be satisfied, so the optimal value is $\frac23$. Therefore, by Lemma~\ref{lemGSameSatisfiability}, $\opt(\mathcal{G}(U_1)) = 1$ and $\opt(\mathcal{G}(U_2)) = \frac23$. However, $\mathcal{G}(U_1)$ and $\mathcal{G}(U_2)$ are $C^2$-equivalent. Let $X$ denote the common variable set of both instances. Duplicator's strategy in the 2-pebble bijective game is to always give a bijection $f: X \to X$ (from the universe of $\mathcal{G}(U_1)$ to the universe of $\mathcal{G}(U_2)$) with the following property:
\begin{equation}\label{equKEquals2PreserveProp}
\txt{For all $v \in V(H)$, there exists $g^*(v) \in A$ such that, for any $g \in A$, $f(x_v^g) = x_v^{g + g^*(v)}$.}
\end{equation}

So, in any given round, Duplicator's bijection is completely determined by a map $g^*: V(H) \to A$. If there is one pebble pair on the board on $(x_v^{g_1}, x_v^{g_2})$, as there is in Figure~\ref{figKEquals2Example}, then for each neighbor $u$ of $v$ in $H$, $g^*(u)$ is uniquely determined to be the value which makes the bijection $f$ an isomorphism across all constraints involving vertices $u$ and $v$. In the figure, $g^*(u) = 1 \in \ff_2$ and $g^*(v) = g^*(w) = 0$. The bijection $f$ is depicted by the blue dotted lines; notice the swap at the vertices involving $u$ corresponding to $g^*(u) = 1$. No matter which of the six locations Spoiler places the pebble pair, one can verify that any relations between two pebbled vertices will be of the same kind.

\ipnc{.39}{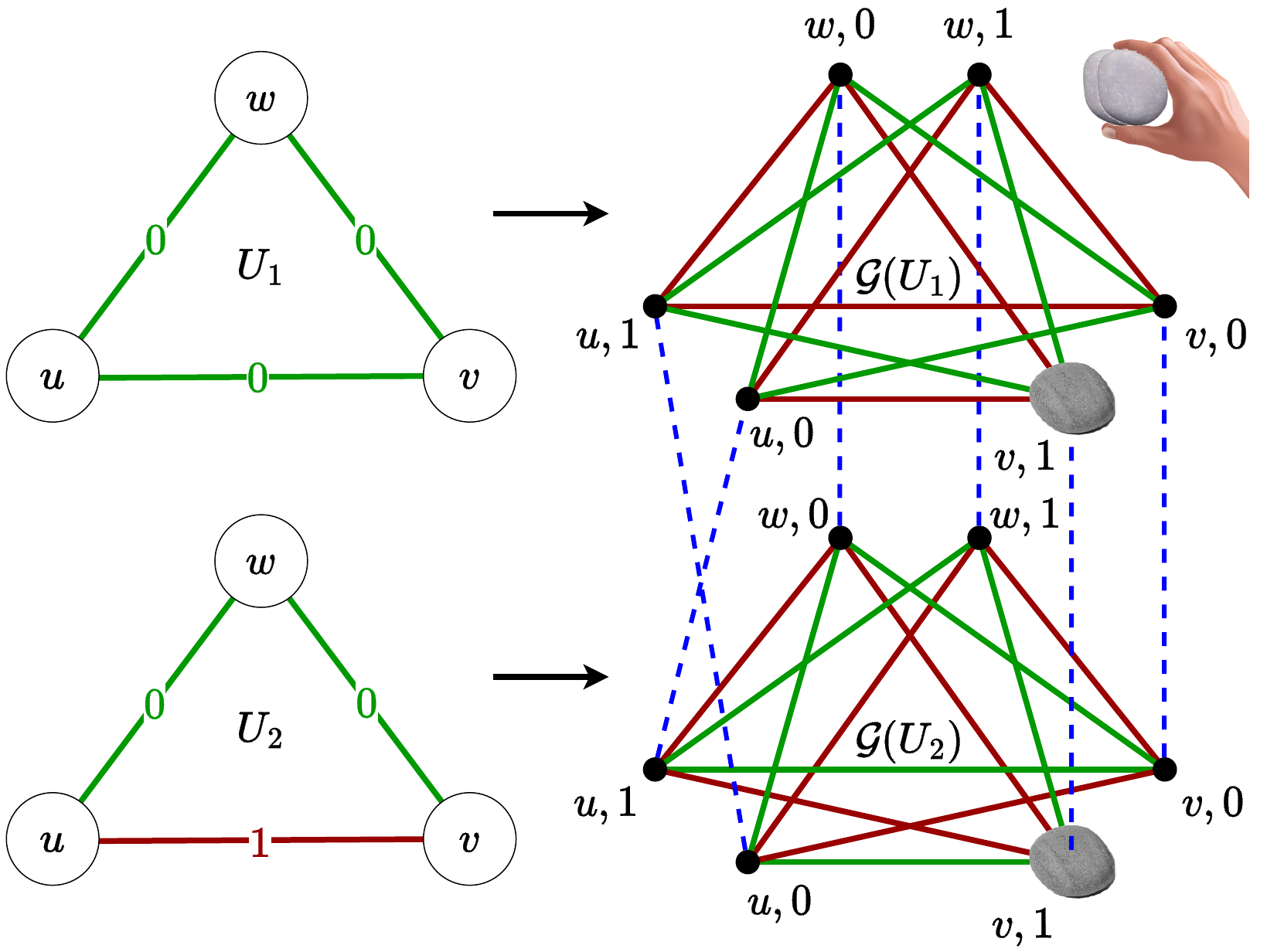}{\label{figKEquals2Example} Two \gug instances $U_1$ and $U_2$ with group $\ff_2$, and Duplicator's bijection at some stage in the $k$-pebble bijective game between $\mathcal{G}(U_1)$ and $\mathcal{G}(U_2)$, just after Spoiler has picked up a pair of pebbles.}

Preserving property (\ref{equKEquals2PreserveProp}) is a more general idea that will be used again in Section~\ref{secMainProof}. We do not go into the general proof, but this strategy enables Duplicator to win the 2-pebble bijective game between $\mathcal{G}(U_1)$ and $\mathcal{G}(U_2)$ as long as $U_1$ and $U_2$ are defined over the same underlying graph and that graph is simple.  It is possible to construct \gug instances $U_2$ over simple graphs with arbitrarily low satisfiability $\delta > 0$ (for example, random constraints on large complete graphs will have low satisfiability with high probability). Turning all constraints into identity constraints, we can then obtain $U_1$ such that $\opt(\mathcal{G}(U_1)) = 1$, $\opt(\mathcal{G}(U_2)) = \delta$, but $\mathcal{G}(U_1) \equiv_{C^2} \mathcal{G}(U_2)$. Thus, for $k = 2$, the objective has been accomplished: we have a $(1, \delta)$-inapproximability gap for $C^2$-definable algorithms, for arbitrarily low $\delta$.

\subsection{The challenge of $k \texorpdfstring{\geq}{≥} 3$}\label{subChallengeOfKEquals3}

To extend this result to an inapproximability gap for FPC-definable algorithms, we need $C^k$-equivalence for all $k$. Unfortunately, we hit a fundamental barrier starting at $k = 3$:

\begin{prop}\label{proUGCompletelySatisfiableInFPC}
	For any positive integer $q$, there is an FPC sentence $\phi$ expressing the property that a $\textup{\UG}(q)$ instance (encoded as a $\tauuug$-structure) is completely satisfiable. Furthermore, $\phi$ can be converted into a $C^3$ sentence.
\end{prop}
\begin{proof}
	For each fixed label $i \in [q]$, in the context of a free variable $x$, we define $q$ unary relations $U_{i, 1}, U_{i, 2}, \dots, U_{i, q}$ by simultaneous induction:
	\begin{align*}
	U_{i, i}(y) &\impliedby (x = y)\\
	U_{i, j}(y) &\impliedby \exists z \bigvee_{\pi: [q] \to [q]} \left(P_\pi(y, z) \wedge U_{i, \pi(j)}(z)\right)
	\end{align*}
	The meaning of $U_{i, j}(y)$ is that, \emph{given $x$ has label $i$, it is implied by the constraints that $y$ has label $j$}. Thus, $U_{i, i}(x)$ is defined to be true, and whenever a constraint $P_\pi$ holds on a pair of elements $(y, z)$ and we know what the label of $z$ must be, we inductively derive what the label of $y$ must be. We claim that the following sentence expresses the property that a $\textup{\UG}(q)$ instance is completely satisfiable:
	\[\phi \equiv \forall x \bigvee_{i \in [q]} \bigwedge_{\substack{j \in [q]\\j \neq i}} \neg U_{i, j}(x)\]
	The correctness of this formula relies on the observation that a \ug instance is \emph{not} completely satisfiable if and only if there exists an inconsistent cycle of constraints, forcing some variable $x$ to be assigned two different values. If $f$ is a satisfying assignment, then picking $i = f(x)$ must satisfy $\phi$. Conversely, if $\phi$ is satisfied, one can obtain a satisfying assignment by picking one $x$ from each connected component of the underlying graph and one satisfying witness $i$, then assigning labels to every $y$ in that component by taking the unique $j$ such that $U_{i, j}(y)$ holds. (It is not too hard to see that the component being connected implies $j$ exists, and $\phi$ being satisfied implies $j$ is unique).
	
	Using the \emph{Bekic principle} \cite[Lemma 1.4.2]{RudimentsOfMuCalculus} \cite[Lemma 10.9]{FMT}, the simultaneous inductions can be nested within each other in a way that reuses variable names, resulting in least fixed-point (LFP) formulas for each of the $q$ relations, still using only 3 variables ($x$, $y$ and $z$). Thus, $\phi$ can indeed be formally written as an LFP sentence, which is an FPC sentence. Finally, $\phi$ can be translated into a $C^3$ sentence by unwrapping the inductive definitions into infinite disjunctions, which again does not require any additional variables since none of the parameter variables are re-quantified within the inductive definitions, so there will be no clashes.
\end{proof}

This implies that there does not exist an FPC-inapproximability gap of $(1, \delta)$ for any $\delta < 1$ (this is sometimes referred to as ``perfect completeness"), since $\phi$ distinguishes an instance of optimal value 1 from an instance of optimal value less than 1. In terms of the 3-pebble bijective game, Spoiler's winning strategy is to drop the first pebble pair anywhere, then use the other two pebble pairs to traverse an inconsistent cycle of constraints in the unsatisfiable instance.

\subsection{An example where $k = 3$}\label{subKEquals3}

From Duplicator's perspective, the difficulty discussed above stems from the fact that the $g^*$ map from (\ref{equKEquals2PreserveProp}) is uniquely defined from neighboring vertices. To circumvent this obstacle and give Duplicator more choices, we place a \emph{bundle} of parallel constraints between some pairs of variables. Unfortunately, this comes with a price: since at most one constraint from each bundle can be satisfied, we will necessarily have $\opt(\mathcal{G}(U_1)) < 1$. Notice how this avoids the perfect completeness issue from Proposition~\ref{proUGCompletelySatisfiableInFPC}.

Figure~\ref{figK4Example} shows an example of the parallel constraint method where the underlying graph is $H := K_4$. The group used is the additive part of $\ff_2^2$, the 2-dimensional vector space over the 2-element field. We denote its elements as binary strings $\{00, 01, 10, 11\}$, so that the group operation is bitwise XOR.

Since the constraints of $U_1$ come in inconsistent pairs, at most half of them can be satisfied. This upper bound is attainable by assigning all variables to be $00$. Thus, $\opt(U_1) = \frac12$. On the other hand, this assigment satisfies only $\frac{5}{12}$ of the edges in $U_2$, and indeed, it can be verified by exhaustive search that $\opt(U_2) = \frac{5}{12}$.

\ipnc{1.1}{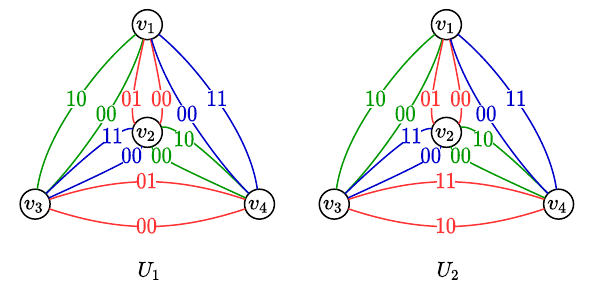}{\label{figK4Example}Two \gug instances $U_1$ and $U_2$ of different optimal values such that $\mathcal{G}(U_1) \equiv_{C^3} \mathcal{G}(U_2)$. The constraints differ only on the bottom two edges.}

Duplicator will follow a similar strategy as in the previous example. Suppose that, at some stage in the 3-pebble bijective game played on $\mathcal{G}(U_1)$ and $\mathcal{G}(U_2)$, there are two pebble pairs already on the board, on variables involving vertices $v_1$ and $v_3$, and Spoiler has just picked up the third pebble pair. This means that $g^*(v_1)$ and $g^*(v_3)$ are already determined, so Duplicator just needs to fill in $g^*(v_2)$ and $g^*(v_4)$. Let us discuss how Duplicator determines $g^*(v_4)$.

Duplicator will set $g^*(v_4)$ so that the map
\[f(x_v^g) := x_v^{g + g^*(v)}\]
is a partial isomorphism among all constraints involving $v_4$ and either $v_1$ or $v_3$. (In particular, this will ensure $f$ is a partial isomorphism between the pebbled variables if Spoiler places the pebble pair down on variables involving $v_4$.) For the constraints involving $v_3$ and $v_4$, $f$ will be a partial isomorphism if, for any $g_3, g_4 \in \ff_2^2$, the following double implication holds:
\begin{align*}
x_{v_3}^{g_3} - x_{v_4}^{g_4} = z &\txt{ is an equation in } \mathcal{G}(U_1)\\
\iff f(x_{v_3}^{g_3}) - f(x_{v_4}^{g_4}) = z &\txt{ is an equation in } \mathcal{G}(U_2).
\end{align*}
Expanding out the definitions of $f$ and $\mathcal{G}$, we can rewrite this condition as
\begin{align*}
(x_{v_3} + g_3) - (x_{v_4} + g_4) = z &\txt{ is an equation in } U_1\\
\iff (x_{v_3} + g_3 + g^*(v_3)) - (x_{v_4} + g_4 + g^*(v_4)) = z &\txt{ is an equation in } U_2.
\end{align*}
Rearranging equations, and using the fact that the equations in $U_2$ between $x_{v_3}$ and $x_{v_4}$ are precisely the same equations in $U_1$ plus a difference of $10$, this becomes
\begin{align*}
x_{v_3} - x_{v_4} = z - g_3 + g_4 &\txt{ is an equation in } U_1\\
\iff x_{v_3} - x_{v_4} = z - g_3 + g_4 - g^*(v_3) + g^*(v_4) + 10 &\txt{ is an equation in } U_1.
\end{align*}
Notice that, in $U_1$, two equations of the form $x_{v_3} - x_{v_4} = z_1$ and $x_{v_3} - x_{v_4} = z_2$ both appear as constraints or both do not appear as constraints if and only if $z_1$ and $z_2$ are the same, or differ by $01$. Thus, $f$ is a partial isomorphism for constraints involving $v_3$ and $v_4$ if and only if
\[g^*(v_4) - g^*(v_3) + 10 \in \{00, 01\}.\]
By an analogous argument, one can derive that $f$ is a partial isomorphism for constraints involving $v_1$ and $v_4$ if and only if
\[g^*(v_4) - g^*(v_1) \in \{00, 11\}.\]
These two constraints on $g^*(v_4)$ can be written as a system of 2 equations with 2 unknowns:
$$\mTwoTwo{1}{0}{1}{1} \vectTwo{g^*(v_4)_1}{g^*(v_4)_2} = \vectTwo{g^*(v_3)_1 + 1}{g^*(v_1)_1 + g^*(v_1)_2}$$
The system always has a solution since the matrix on the left-hand side has full rank, which is ultimately due to the fact that the subspaces $\{00, 01\}$ and $\{00, 11\}$ are orthogonal. Thus, no matter what $g^*(v_1)$ and $g^*(v_3)$ are, we have seen that Duplicator can find a value for $g^*(v_4)$ making $f$ an isomorphism over all constraints along the path $(v_1, v_4, v_3)$. For example, if $g^*(v_3) = 00$ and $g^*(v_1) = 01$, then we set $g^*(v_4) = 10$. Duplicator can use the same process to find a value for $g^*(v_2)$ making $f$ an isomorphism over the path $(v_1, v_2, v_3)$.

\section{Proof of FPC-inapproximability gap}\label{secMainProof}

In comparison to the previous example, the general construction replaces $\ff_2^2$ with $\fm$ and has a bundle of $2^\ell$ constraints between every pair of adjacent vertices, for integers $0 < \ell < m$. A key difference is that it is no longer possible to make $f$ an isomorphism over an arbitrary path of length 2 (such as $(v_1, v_4, v_3)$ in the example) given arbitrary values for the endpoints. However, we show that this \emph{is} possible for all paths in the graph that are sufficiently long (see Lemma~\ref{lemNUGLGRadiusProperty}). This allows Duplicator to win as long as the base graph $H$ has suitably high girth.

\subsection{Construction}\label{subConstruction}

Our construction takes three parameters: $\varepsilon \in (0, \frac12)$, $\delta > 0$, and a positive integer $\ell$. Note that $\delta$ and $\ell$ affect the satisfiability of the \gug instances produced, whereas $\varepsilon$ is an arbitrary constant that could just as well be fixed at $\frac14$. The construction is probabilistic, failing with probability at most $2\varepsilon$. We let
\begin{equation}\label{equChooseD}
d := 2^\ell + 1,
\end{equation}
which will be the degree of the vertices in the underlying graph. The reader may find it helpful to imagine $\ell = 1$ and $d = 3$ throughout this section. Next, we let
\begin{equation}\label{equChooseGamma}
\gamma := 1 - \frac{\left(\soundness + \frac{\delta}{2^\ell}\right)}{\left(\soundness + \delta\right)}.
\end{equation}
Note that $\gamma > 0$. This is just a technical upper bound on the fraction of edges in the graph that have a certain undesirable property (see Lemma~\ref{lemNUGLGMostEdgesGood}). Finally, we pick two very large integers $0 \ll m \ll r$; specifically,
\begin{equation}\label{equChooseM}
m := \left\lceil \frac{1}{\delta} + \left(\frac{2}{\delta d} + 1\right)\ell - \frac{2}{\delta d} \log_2(\varepsilon) \right\rceil,
\end{equation}
and
\begin{equation}\label{equChooseR}
r := \left\lceil\frac{2^{m\ell + 1}(d^{m - \ell} - 1)}{\gamma \varepsilon}\right\rceil + 1.
\end{equation}

The \gug instances we construct all have the additive part of the vector space $\fm$ as the group, so $q := 2^m$. The parameter $m$ is chosen so that $\mathbb{B}_k = \mathcal{G}(U_2)$ will be highly unsatisfiable with high probability. The meaning of the parameter $r$ will be discussed shortly.

For any $k$, let $\widetilde{H} = \widetilde{H}_k$ be a $d$-regular simple graph of girth at least $(k + 1)^2 r$. (Regular graphs of arbitrarily high girth and degree are known to exist; see Lazebnik, Ustimenko and Woldar \cite{DenseGraphsHighGirth}, for example.) The final graph we use for the construction, $H$, will later be obtained from $\widetilde{H}$ by removing certain edges. For every edge $\{v_1, v_2\} \in E(\widetilde{H})$, we independently choose a uniformly random vector $b(v_1, v_2) = b(v_2, v_1) \in \fm$ and a uniformly random $\ell$-dimensional subspace\footnote{What this means is, randomly choose a set of $\ell$ linearly independent vectors and take the span. Choose the first vector uniformly at random from $\fm \setminus \{0\}$, then choose each subsequent vector uniformly at random from the subset of $\fm$ that is not in the span of the previously chosen vectors. } $Z(v_1, v_2) = Z(v_2, v_1) \subseteq \fm$. Say that an edge $e \in E(\widetilde{H})$ is \emph{good} if, for all paths $(v_0, \seq{v}{r})$ of length $r$ such that $e = \{v_0, v_1\}$, the set
\[\bigcup_{1 \leq i \leq r} Z(v_{i - 1}, v_i)\]
spans $\fm$ (recall that paths are never allowed to repeat edges). Edges of $\widetilde{H}$ which are not good edges are called \emph{bad} edges. We will shortly argue that, with high probability, nearly all edges are good edges.

Let $H = H_k$ be the graph with vertex set $V(H) := V(\widetilde{H})$ and edge set
\[E(H) := \{e \in E(\widetilde{H}) \suchthat \txt{$e$ is a good edge}\}.\]

We define \gug instances $U_1$, $\widetilde{U}_1$, $U_2$ and $\widetilde{U}_2$ using the additive group structure on $\fm$. The variable sets of all four instances are
\[\{x_v \suchthat v \in V(H)\}.\]
For every edge $\{v_1, v_2\} \in E(\widetilde{H})$, $\widetilde{U}_1$ and $\widetilde{U}_2$ have a bundle of $2^\ell$ constraints between the corresponding variables. In $\widetilde{U}_1$, the constraints are
\[\{x_{v_1} - x_{v_2} = z \suchthat z \in Z(v_1, v_2)\},\]
whereas in $\widetilde{U}_2$, the constraints are
\[\{x_{v_1} - x_{v_2} = z + b(v_1, v_2) \suchthat z \in Z(v_1, v_2)\}.\]
Finally, $U_1$ and $U_2$ are obtained from $\widetilde{U}_1$ and $\widetilde{U}_2$ by removing all constraints on pairs of variables corresponding to bad edges, i.e., with constraints defined in the exact same way as $\widetilde{U}_1$ and $\widetilde{U}_2$, but only for edges $\{v_1, v_2\} \in E(H)$. Our pair of $C^k$-equivalent instances will be $\mathbb{A}_k := \mathcal{G}(U_1)$, and $\mathbb{B}_k := \mathcal{G}(U_2)$.

The example from Section~\ref{subKEquals3} is a special case of this construction where $k := 3$, $H := K_4$, $m := 2$, $\ell := 1$, and $r := 2$. The 1-dimensional $Z$-subspaces consist of the vectors drawn on each edge in $U_1$ in Figure~\ref{figK4Example}, with $b(v_3, v_4) := 10$ and all other $b$ vectors $00$.

\subsection{Satisfiability properties}\label{subProperties}

In the next section we prove that $\mathcal{G}(U_1)$ and $\mathcal{G}(U_2)$ are always $C^k$-equivalent. But first, we now analyze the optimal values of these two instances.

\begin{lem}\label{lemNUGLGCompleteness}
	The satisfiability of $U_1$ (and thus of $\mathcal{G}(U_1)$) is exactly $\frac{1}{2^\ell}$.
\end{lem}
\begin{proof}
	Since the constraints in each bundle of $U_1$ are pairwise incompatible with each other, at most one constraint can be satisfied from each bundle, so the total satisfiability is at most the bundle size $\frac{1}{2^\ell}$. The assignment $x_v := 0$ attains this bound by satisfying the $z = 0$ constraint in each bundle (every subspace $Z(v_1, v_2)$ must contain $z = 0$). The claim about $\mathcal{G}(U_1)$ follows from Lemma~\ref{lemGSameSatisfiability}.
\end{proof}

We next argue that, with high probability, $\mathcal{G}(U_2)$ has a significantly lower optimal value. The analysis requires many steps, so we break the proof up over several lemmas. Say that a bundle of constraints is ``satisfied" if at least one constraint in the bundle is satisfied. We begin by showing in Lemma~\ref{lemNUGLGMostEdgesGood} that, with high probability, most edges of $\widetilde{H}$ are good edges, in which case the satisfiability of $\widetilde{U}_2$ closely approximates the satisfiability of $U_2$. To analyze the expected satisfiability of $\widetilde{U}_2$, in Lemma~\ref{lemSoundnessU2Tilde} we observe that any assignment satisfying significantly more than a
\[\soundness = \frac{1}{2^\ell} \cdot \frac2d\]
fraction of constraints must completely satisfy all of the bundles over some subgraph of $\widetilde{H}$ containing significantly more than a $\frac2d$-fraction of the edges of $\widetilde{H}$. Such a subgraph must contain many more edges than vertices, so it is highly improbable that a random assignment will completely satisfy it. We then apply a union bound over all assignments and all subgraphs to conclude that, with high probability, no assignment satisfies more than a $\soundness$ fraction of constraints.

In order to show that most edges are good edges, we will require the following technical lemma.

\begin{lem}\label{lemWierdIntegerBound}
	For any $d \geq 3$ and positive integer $n$,
	\[\frac{d^n - d^{n - 1}}{(d - 1)^n} + d^{n - 1} - 1 \leq \frac{d^n - 1}{d - 1}.\]
\end{lem}

\begin{proof}
	When $n = 1$, both sides are 1. When $n = 2$, the inequality reads
	\[\frac{d}{d - 1} + d - 1 \leq d + 1,\]
	which is clearly true for $d \geq 3$. So assume $n \geq 3$. Then, using the fact that $d \geq 3$, we have
	\begin{align*}
	(d - 1)^{n - 1} &\geq (d - 1)^2 = (d - 1)(d - 1)\\
	&\geq (d - 1)(3 - 1) = d + (d - 2)\\
	&\geq d + (3 - 2) \geq d.
	\end{align*}
	It follows that
	\[d^n \leq d^{n - 1}(d - 1)^{n - 1}.\]
	Multiplying this inequality by the equation $1 = d - (d - 1)$, we have
	\[d^n \leq d^{n}(d - 1)^{n - 1} - d^{n - 1}(d - 1)^n.\]
	Adding the inequalities $0 \leq (d - 1)^n$ and $(d - 1)^{n - 1} \leq d^{n - 1}$, we obtain
	\[d^n + (d - 1)^{n - 1} \leq d^{n}(d - 1)^{n - 1} - d^{n - 1}(d - 1)^n + (d - 1)^n + d^{n - 1}.\]
	Rearranging,
	\[d^n - d^{n - 1} + d^{n - 1}(d - 1)^n - (d - 1)^n \leq d^{n}(d - 1)^{n - 1} - (d - 1)^{n - 1}.\]
	Finally, dividing both sides by $(d - 1)^n$ yields the desired inequality.
\end{proof}

\begin{lem}\label{lemNUGLGMostEdgesGood}
	The probability that more than a $\gamma$ fraction of the edges of $\widetilde{H}$ are bad edges is at most $\varepsilon$.
\end{lem}

This is our most difficult lemma, so we first give some intuition behind the proof. Consider the special case where $m = 2$, $\ell = 1$ and $d = 3$. In other words, $\widetilde{H}$ is a 3-regular graph with each edge labeled by one of the three one-dimensional subspaces of $\ff_2^2$. Observe that $d = 2^\ell + 1$, as in (\ref{equChooseD}). Choose a random edge $e^* = \{u, v\} \in E(\widetilde{H})$. Recall that the definition of a good/bad edge depends on the parameter $r$ from our construction, which is the length of paths to consider. We imagine starting with $r = 1$ and increasing $r$ until (hopefully) $e^*$ becomes a good edge. Specifically, for any positive integer $r$, we let $X_r$ be the random variable defined as the number of paths of length $r$ starting from $e^*$ in which the union of the subspaces of $\ff_2^2$ along $p$ does \emph{not} span $\ff_2^2$. We call these ``bad'' paths. At $r = 1$, there are two paths to consider: $(u, v)$ and $(v, u)$. Since $\ell = 1$, $Z(u, v) = Z(v, u)$ is a 1-dimensional subspace of $\ff_2^2$, so both of these paths fail to span the entire space and are counted as bad paths. Thus, $X_1 = 2$ with probability 1. To compute $X_2$, observe that, since the graph is 3-regular, both of these bad paths of length 1 can be extended in exactly 2 different ways to be paths of length 2. So suppose the path $(u, v)$ can be extended to either $(u, v, w_1)$ or $(u, v, w_2)$. If $Z(v, w_1) \subseteq Z(u, v)$ (which happens with probability $\frac13 < \frac12$ since there are three nonzero vectors in $\ff_2^2$ and only one lies within the subspace), then the extension $(u, v, w_1)$ will still fail to generate the space, so it will still be counted in $X_2$. Thus, with probability more than $\frac12$, the extension \emph{will} generate the space, and not be counted in $X_2$. In this case we say that the path becomes a ``good'' path. The same holds for $(u, v, w_2)$. Thus, as we extend the bad path $(u, v)$ by 1 vertex, we have created 2 new potentially bad paths, but each one immediately becomes a good path with probability $> \frac12$, so the expected number of bad extensions is less than 1. Similarly, the expected number of bad extensions of $(v, u)$ is less than 1 as well. More generally, this argument shows that the expected value of $X_r$ is strictly decreasing in $r$. Thus, with high probability, $X_r$ will eventually hit zero, at which point we can conclude that $e^*$ has no bad paths, so it is a good edge.
	
This special case contains the main idea of the proof. Equation (\ref{equChooseD}) is the key, ensuring that the expected value of $X_r$ is decreasing: at each path extension, $d - 1 = 2^\ell$ new potentially bad paths are created, and each becomes a good path with probability $> \frac{1}{2^\ell}$. The specific choice of $r$ in (\ref{equChooseR}) is just a bound on how large $r$ may need to be to ensure $X_r = 0$ with high probability.

\begin{proof}[Proof of Lemma~\ref{lemNUGLGMostEdgesGood}]
	Let $e^*$ be an edge of $\widetilde{H}$ chosen uniformly at random. For each $i \in [r]$, let $P_i$ be the set of paths of length $i$ whose first two vertices are the endpoints of $e^*$. For any path $p = (v_0, \seq{v}{n})$ in $\widetilde{H}$, we say that the \emph{dimension} of $p$, denoted $\dim(p)$, is the dimension of the vector space
	\[Z(p) := \txt{span}\left(\bigcup_{i \in [n]} Z(v_{i - 1}, v_i)\right).\]
	We define random variables $\seq{X}{r}$ as follows:
	\begin{equation*}\label{equXIDef}
	X_i := \sum_{p \in P_i} \left(d^{m - \dim(p)} - 1\right).
	\end{equation*}
	(Note that this is more general than the definition in the proof sketch, as it counts subspaces of different dimensions with different weights.)
	
	With probability 1,
	\begin{equation}\label{equX1Equals}
	X_1 = 2(d^{m - \ell} - 1),
	\end{equation}
	since there are only 2 paths of length 1 to consider (depending on which direction $e^*$ is traversed), and each has dimension $\ell$. By definition, $X_i \geq 0$ for all $i \in [r]$, and $e^*$ is a good edge if and only if $X_r = 0$. Thus, our goal is to upper-bound $\Pr[X_r \geq 1]$.
	
	\begin{clm}\label{claExpectationDecreasing}
		For any $2 \leq i \leq r$ and $x \geq 1$,
		\begin{equation*}\label{equExpectationStrictlyDecreasing}
		\ee[X_i \suchthat X_{i - 1} = x] \leq x - \frac{1}{2^{m\ell}}.
		\end{equation*}
	\end{clm}
	
	\begin{proof}[Proof of Claim.]
		We prove that this holds even after additionally conditioning on any tuple of realizations $R$ of the random subspaces along all paths in $P_{i - 1}$ that is consistent with $X_{i - 1} = x$. Note that conditioning on these realizations also fixes the value of $\dim(p_0)$ for any $p_0 \in P_{i - 1}$.
		
		We can split a path $p = (v_0, \seq{v}{i})$ into two subpaths $p_0 := (v_0, \seq{v}{i - 1})$ and $p_1 := (v_{i - 1}, v_i)$. When this happens, we say $p_1 \succ p_0$ and $p_0 + p_1 = p$. Using this decomposition, we obtain the following bound on $X_i$:
		\begin{align*}
		X_i &= \sum_{p_0 \in P_{i - 1}} \sum_{p_1 \succ p_0} (d^{m - \dim(p_0 + p_1)} - 1) \leq \sum_{p_0 \in P_{i - 1}} \sum_{p_1 \succ p_0} \twocases{\txt{ if } Z(p_1) \subseteq Z(p_0)}{d^{m - \dim(p_0)} - 1}{\txt{ otherwise}}{d^{m - \dim(p_0) - 1} - 1}
		\end{align*}
		Note that, since the girth of $\widetilde{H}$ is greater than $r$, as we expand the radius of edges we are examining, there are never any cycles. Thus, the events that each $Z(p_1) \subseteq Z(p_0)$ holds are independent, and also independent of the previously revealed subspaces $R$. For convenience, let us denote
		\[f(p_0, p_1) := \twocases{\txt{ if } Z(p_1) \subseteq Z(p_0)}{d^{m - \dim(p_0)} - 1}{\txt{ otherwise}}{d^{m - \dim(p_0) - 1} - 1},\]
		so that
		\[X_i \leq \sum_{p_0 \in P_{i - 1}} \sum_{p_1 \succ p_0} f(p_0, p_1).\]
		Then
		{\allowdisplaybreaks
		\begin{align*}
		\ee[X_i \suchthat R] &\leq \ee\left[\sum_{p_0 \in P_{i - 1}} \sum_{p_1 \succ p_0} f(p_0, p_1) \ \bigg{|}\ R \right]\\
		&= \sum_{p_0 \in P_{i - 1}} \sum_{p_1 \succ p_0} \ee\left[f(p_0, p_1) \suchthat R\right]\\
		&= \sum_{\substack{p_0 \in P_{i - 1} \\ n := m - \dim(p_0) \geq 1}} \sum_{p_1 \succ p_0} \ee\left[f(p_0, p_1) \suchthat R\right] \stextn{since $f(p_0, p_1) = 0$ for all $p_0$ of dimension $m$}\\
		&= \sum_{\substack{p_0 \in P_{i - 1} \\ n \geq 1}} \sum_{p_1 \succ p_0} \left(\begin{array}{ll} \Pr[Z(p_1) \subseteq Z(p_0) \suchthat R] (d^n - 1) + \txt{}\\ (1 - \Pr[Z(p_1) \subseteq Z(p_0) \suchthat R]) (d^{n - 1} - 1) \end{array} \right)\\
		&= \sum_{\substack{p_0 \in P_{i - 1} \\ n \geq 1}} \sum_{p_1 \succ p_0} \left(\Pr[Z(p_1) \subseteq Z(p_0) \suchthat R] (d^n - d^{n - 1}) + d^{n - 1} - 1\right)\\
		&\leq \sum_{\substack{p_0 \in P_{i - 1} \\ n \geq 1}} \sum_{p_1 \succ p_0} \left(\left(\frac{2^{\dim(p_0)} - 1}{2^m - 1}\right)^\ell\left(d^n - d^{n - 1}\right) + d^{n - 1} - 1\right).
		\end{align*}}
		This inequality follows from the observation that, if we choose a vector from $\fm \setminus \{0\}$ uniformly at random, it lies within $Z(p_0)$ with probability $\frac{2^{\dim(p_0)} - 1}{2^m - 1}$. And we may think of choosing a random $\ell$-dimensional subspace as repeatedly sampling nonzero vectors independently and uniformly at random until the dimension of all vectors sampled so far is $\ell$. In order to have $Z(p_1) \subseteq Z(p_0)$, it is necessary for all of these vectors to lie within $Z(p_0)$, so in particular, the first $\ell$ vectors must be contained in $Z(p_0)$. Thus, the probability that $Z(p_1) \subseteq Z(p_0)$ holds is at most the probability that just the first $\ell$ vectors are contained in $Z(p_0)$, which is $\left((2^{\dim(p_0)} - 1)/(2^m - 1)\right)^\ell$.
		
		Note that, when $n := m - \dim(p_0) \geq 1$,
		\begin{align*}
		\left(\left(2^{\dim(p_0) - m}\right)^\ell - \left(\frac{2^{\dim(p_0)} - 1}{2^m - 1}\right)^\ell\right)&\left(d^n - d^{n - 1}\right) \geq \left(2^{\dim(p_0) - m}\right)^\ell - \left(\frac{2^{\dim(p_0)} - 1}{2^m - 1}\right)^\ell\\
		&= \frac{(2^{\dim(p_0)}(2^m - 1))^\ell - ((2^{\dim(p_0)} - 1)2^m)^\ell}{(2^m(2^m - 1))^\ell}\\
		&\geq \left(\frac{2^{m\ell} - 2^{\dim(p_0)\ell}}{(2^m(2^m - 1))^\ell}\right)\\
		&= \left(\frac{1 - 2^{\ell(\dim(p_0) - m)}}{(2^m - 1)^\ell}\right)\\
		&= \left(1 - 2^{-\ell n}\right)\left(2^m - 1\right)^{-\ell}\\
		&\geq \left(1 - 2^{-\ell}\right)2^{-m\ell} \snc{n \geq 1}\\
		&\geq 2^{-\ell}2^{-m\ell} \snc{\ell \geq 1}\\
		&= \frac{1}{2^{(m + 1)\ell}}.
		\end{align*}
		Therefore,
		\begin{align*}
		\left(\frac{2^{\dim(p_0)} - 1}{2^m - 1}\right)^\ell\left(d^n - d^{n - 1}\right) + d^{n - 1} - 1 &\leq \left(2^{\dim(p_0) - m}\right)^\ell\left(d^n - d^{n - 1}\right) \push + d^{n - 1} - 1 - \frac{1}{2^{(m + 1)\ell}}\\
		&= \left(\frac{1}{(2^{\ell})^n}\right)\left(d^n - d^{n - 1}\right) + d^{n - 1} - 1 - \frac{1}{2^{(m + 1)\ell}}\\
		&= \left(\frac{d^n - d^{n - 1}}{(d - 1)^n} + d^{n - 1} - 1\right) - \frac{1}{2^{(m + 1)\ell}} \stextn{from the definition of $d$ in (\ref{equChooseD})}\\
		&\leq \frac{d^n - 1}{d - 1} - \frac{1}{2^{(m + 1)\ell}} \stext{from Lemma~\ref{lemWierdIntegerBound}}.
		\end{align*}
		
		Putting this all together, we have
		\begin{align*}
		\ee[X_i \suchthat R] &\leq \sum_{\substack{p_0 \in P_{i - 1} \\ n := m - \dim(p_0) \geq 1}} \sum_{p_1 \succ p_0} \left(\left(\frac{2^{\dim(p_0)} - 1}{2^m - 1}\right)^\ell\left(d^n - d^{n - 1}\right) + d^{n - 1} - 1\right)\\
		&\leq \sum_{\substack{p_0 \in P_{i - 1} \\ n := m - \dim(p_0) \geq 1}} \sum_{p_1 \succ p_0} \left(\frac{d^n - 1}{d - 1} - \frac{1}{2^{(m + 1)\ell}}\right)\\
		&= \sum_{\substack{p_0 \in P_{i - 1} \\ n := m - \dim(p_0) \geq 1}} \left(d^n - 1 - \frac{d - 1}{2^{(m + 1)\ell}}\right)\\
		&= \sum_{\substack{p_0 \in P_{i - 1} \\ n := m - \dim(p_0) \geq 1}} \left(d^n - 1 - \frac{1}{2^{m\ell}}\right) \stext{from (\ref{equChooseD})}\\
		&\leq \left(\sum_{\substack{p_0 \in P_{i - 1} \\ n := m - \dim(p_0) \geq 1}} d^n - 1\right) - \frac{1}{2^{m\ell}}\\
		&= \left(\sum_{p_0 \in P_{i - 1}} \left(d^{m - \dim(p_0)} - 1\right)\right) - \frac{1}{2^{m\ell}} \stextn{since $n = 0$ terms contribute zero to the sum}\\
		&= x - \frac{1}{2^{m\ell}},
		\end{align*}
		where in the final inequality we have used the assumption that $R$ is consistent with $X_{i - 1} = x \geq 1$, so there is at least one term in the sum, i.e., at least one bad path. Thus, Claim~\ref{claExpectationDecreasing} holds.
	\end{proof}
	
	Returning to the proof of Lemma~\ref{lemNUGLGMostEdgesGood}, we now suppose toward a contradiction that $\Pr[X_r \geq 1] > \gamma \varepsilon$. Since $X_i = 0$ implies that $X_j = 0$ for all $i \leq j$, it follows that, for all $i \in \{1, 2, 3, \dots, r - 1\}$,
	\begin{equation}\label{equAllXProbablyNonzero}
		\Pr[X_i \geq 1] > \gamma \varepsilon.
	\end{equation}
	
	For $2 \leq i \leq r$, observe that
	\begin{align*}
	\ee[X_i] &= \sum_{x \geq 0} \Pr[X_{i - 1} = x] \ee[X_i \suchthat X_{i - 1} = x]\\
	&= \Pr[X_{i - 1} = 0] \ee[X_i \suchthat X_{i - 1} = 0] + \sum_{x \geq 1} \Pr[X_{i - 1} = x] \ee[X_i \suchthat X_{i - 1} = x]\\
	&\leq 0 + \sum_{x \geq 1} \Pr[X_{i - 1} = x] \left(x - \frac{1}{2^{m\ell}}\right) \stext{by Claim~\ref{claExpectationDecreasing}}\\
	&= \ee[X_{i - 1}] - \frac{\Pr[X_{i - 1} \geq 1]}{2^{m\ell}}.
	\end{align*}
	Iteratively applying this inequality for $i = r, r - 1, \dots, 3, 2$, we conclude that
	\begin{align*}
	\ee[X_r] &\leq \ee[X_1] - \sum_{i = 1}^{r - 1} \frac{\Pr[X_{i} \geq 1]}{2^{m\ell}}\\
	&< 2(d^{m - \ell} - 1) - \sum_{i = 1}^{r - 1} \frac{\gamma \varepsilon}{2^{m\ell}} \stext{from (\ref{equX1Equals}) and (\ref{equAllXProbablyNonzero})}\\
	&= 2(d^{m - \ell} - 1) - \frac{(r - 1) \gamma \varepsilon}{2^{m\ell}}\\
	&\leq 0 \stext{from the definition of $r$ in (\ref{equChooseR})},
	\end{align*}
	which is a contradiction, since $X_r$ can never be negative.
	
	Hence,
	\[\Pr[e^* \txt{ is a bad edge}] = \Pr[X_r \geq 1] \leq \gamma \varepsilon,\]
	so the expected fraction of bad edges is at most $\gamma\varepsilon$. Markov's inequality states that, for a nonnegative real-valued random variable $Y$ and any $a > 0$, the probability that $Y \geq a$ is at most $\frac{\ee[Y]}{a}$. Applying Markov's inequality with $Y$ as the fraction of bad edges and $a := \gamma$, we conclude that the probability the fraction of bad edges is greater than $\gamma$ is thus at most $\frac{\gamma\varepsilon}{\gamma} = \varepsilon$.
\end{proof}

\begin{lem}\label{lemSoundnessU2Tilde}
	With probability at least $1 - \varepsilon$, the satisfiability of $\widetilde{U}_2$ is less than
	\[\left(1 - \gamma\right) \left(\soundness + \delta\right).\]
\end{lem}

\begin{proof}
	Call a subgraph of $\widetilde{H}$ \emph{large} if it contains all of the vertices of $\widetilde{H}$ and at least a $(\frac{2}{d} + \delta)$ fraction of the edges of $\widetilde{H}$. Recall the definition of $d$ and $\gamma$ in respective Equations (\ref{equChooseD}) and (\ref{equChooseGamma}). If an assignment satisfies at least a
	\[\left(1 - \gamma\right) \left(\soundness + \delta\right) = \left(\soundness + \frac{\delta}{2^\ell}\right)\]
	fraction of the constraints in $\widetilde{U}_2$, then, since at most one constraint from each bundle can be satisfied, and there are $2^\ell$ constraints in each bundle, the assignment must satisfy the bundles of at least a
	\[2^\ell \left(\soundness + \frac{\delta}{2^\ell}\right) = \frac{1}{2^{\ell - 1} + \frac{1}{2}} + \delta = \frac{2}{d} + \delta\]
	fraction of edges. (Recall that to ``satisfy a bundle" is to satisfy one constraint in the bundle.) In other words, there must be some large subgraph on which the assignment completely satisfies all of the bundles. Therefore, by the union bound, which says that the probability of the union of multiple events is at most the sum of the individual probabilities of each of the events, we conclude that the probability that some assignment satisfies at least a $\left(1 - \gamma\right) \left(\soundness + \delta\right)$ fraction of constraints is at most the number of large subgraphs times the probability that some assignment satisfies all bundles on a given large subgraph. Supposing that $\widetilde{H}$ (which is $d$-regular) has $n$ vertices, the number of large subgraphs of $\widetilde{H}$ is bounded by $2^{\frac{nd}{2}}$, the total number of subgraphs containing all vertices of $\widetilde{H}$. To bound the probability that some assignment satisfies all bundles on a given large subgraph, we first note that the probability of an arbitrary assignment satisfying all bundles is at most
	\[\left(\frac{2^\ell}{2^m}\right)^{\frac{nd}{2}\left(\frac{2}{d} + \delta\right)} = 2^\wedge\hspace{-.3em}\left[(\ell - m)\frac{nd}{2}\left(\frac{2}{d} + \delta\right)\right],\]
	since there are at least $\frac{nd}{2}\left(\frac{2}{d} + \delta\right)$ edges in any large subgraph, and each edge has its bundle satisfied independently with probability $\frac{2^\ell}{2^m}$. Thus, applying the union bound again, the probability that some assignment satisfies all bundles on a given large subgraph is at most $(2^m)^n$ times this quantity. So, in total, the probability that some assignment satisfies at least a $\left(1 - \gamma\right) \left(\soundness + \delta\right)$ fraction of constraints is at most
	\begin{align*}
	&\ 2^\wedge\hspace{-.3em}\left[\frac{nd}{2} + mn + (\ell - m)\frac{nd}{2}\left(\frac{2}{d} + \delta\right)\right]\\
	=&\ 2^\wedge\hspace{-.3em}\left[n\left(\frac{d}{2} + \frac{\delta d}{2}\left(\frac{2}{\delta d} + 1\right)\ell - \frac{\delta d}{2} \cdot m\right)\right]\\
	\leq&\ 2^\wedge\hspace{-.3em}\left[n \hspace{.16em} \left(\frac{d}{2} + \frac{\delta d}{2}\left(\frac{2}{\delta d} + 1\right)\ell - \frac{\delta d}{2} \cdot \left(\frac{1}{\delta} + \left(\frac{2}{\delta d} + 1\right)\ell - \frac{2}{\delta d} \log_2(\varepsilon)\right)\right)\right] \stext{from (\ref{equChooseM})}\\
	=&\ 2^{n \log_2(\varepsilon)}\\
	\leq&\ 2^{\log_2(\varepsilon)} \bcause{\varepsilon < 1 \implies \log_2(\varepsilon) < 0, \txt{ and } n > 0}\\
	=&\ \varepsilon.
	\end{align*}
	Therefore, the probability that no assignment satisfies at least this fraction of constraints is at least $1 - \varepsilon$.
\end{proof}

The final step is to observe that $\widetilde{U}_2$ being highly unsatisfiable implies $U_2$ is also highly unsatisfiable.

\begin{lem}\label{lemNUGLGSoundnessU2}
	With probability at least $1 - 2\varepsilon$, the satisfiability of $U_2$ (and thus of $\mathcal{G}(U_2)$) is less than
	\[s := \soundness + \delta.\]
\end{lem}

\begin{proof}
	By Lemma~\ref{lemNUGLGMostEdgesGood}, the probability that less than a $(1 - \gamma)$ fraction of edges of $\widetilde{H}$ are good edges is at most $\varepsilon$. By Lemma~\ref{lemSoundnessU2Tilde}, the probability that $\widetilde{U}_2$ is $(1 - \gamma) s$-satisfiable is at most $\varepsilon$ as well. By the union bound, the probability that either of these two events occurs is at most $2\varepsilon$, so the probability that neither event occurs is at least $1 - 2\varepsilon$. So it suffices to prove that, whenever at least a $(1 - \gamma)$ fraction of the edges of $\widetilde{H}$ are good edges, if $\widetilde{U}_2$ is not $(1 - \gamma) s$-satisfiable, then $U_2$ is not $s$-satisfiable.
	
	We prove the contrapositive, i.e., that if at least an $s$ fraction of constraints are satisfiable in $U_2$, then at least a $(1 - \gamma)s$ fraction of constraints are satisfiable in $\widetilde{U}_2$. Suppose that $\widetilde{U}_2$ has a total of $c$ constraints. Then $U_2$ has at least $(1 - \gamma)c$ constraints. So if at least an $s$ fraction of constraints are satisfiable in $U_2$, it means that at least $s(1 - \gamma)c$ constraints of $U_2$ are satisfied by some assignment $x_v$. Since $U_2$ and $\widetilde{U}_2$ have the same variable set, and all of the constraints of $U_2$ are also constraints of $\widetilde{U}_2$, it follows that $x_v$ must satisfy $s(1 - \gamma)c$ constraints of $\widetilde{U}_2$ as well, that is, at least an $s(1 - \gamma)$ fraction of constraints.
\end{proof}

\subsection{Proof of $C^k$-equivalence}\label{subCkEquivalenceProof}

The following lemma is a generalization of Duplicator's strategy for determining $g^*(v_4)$ in the example from Section~\ref{subKEquals3}. Recall that $H$ is the underlying graph of both instances, consisting of only good edges.

\begin{lem}\label{lemNUGLGRadiusProperty}
	Let $p = (v_0, \seq{v}{n})$ be a path in $H$ of length $n \geq r$. Given any values in $\fm$ for $g^*(v_0)$ and $g^*(v_n)$, it is possible to extend $g^*$ to all of the intermediate vertices of $p$ so that the map $f(x_v^g) := x_v^{g + g^*(v)}$ is a partial isomorphism between $\mathcal{G}(U_1)$ and $\mathcal{G}(U_2)$ over the substructure with universe $\{x_v^g \suchthat v \in p,\ g \in \fm\}$ and relations involving consecutive vertices in $p$.
\end{lem}
\begin{proof}
	Since $H$ contains only good edges and $p$ has length at least $r$, there exists a set of vectors
	\[B \subseteq \bigcup_{i \in [n]} Z(v_{i - 1}, v_i)\]
	forming a basis of $\fm$. Write $h(i)$ for the number of basis vectors in $B$ taken from $Z(v_{i - 1}, v_i)$, and denote these vectors by
	\[B = \bigcup_{i \in [n]} \{z_{i,j} \suchthat j \in [h(i)]\},\]
	where each $z_{i,j}$ is in $Z(v_{i - 1}, v_i)$. Since $B$ is a basis, there exist coefficients $c_{i, j}$ such that
	\begin{equation}\label{equGStarBasis}
	g^*(v_0) - g^*(v_n) - \sum_{i \in [n]} b(v_{i - 1}, v_{i}) = \sum_{i \in [n]} \sum_{j \in [h(i)]} c_{i, j}z_{i, j}.
	\end{equation}
	For each $i$ in order from $1$ to $n$, inductively define
	\[g^*(v_i) := g^*(v_{i - 1}) - \sum_{j \in [h(i)]} c_{i, j}z_{i, j} - b(v_{i - 1}, v_i).\]
	Note that, by expanding the inductive definition for $g^*(v_n)$, we have
	\begin{align*}
	g^*(v_n) &= g^*(v_{n - 1}) - \sum_{j \in [h(n)]} c_{n, j}z_{n, j} - b(v_{n - 1}, v_n)\\
	&= g^*(v_{n - 2}) - \sum_{j \in [h(n - 1)]} c_{n - 1, j}z_{n - 1, j} - b(v_{n - 2}, v_{n - 1}) \push\hspace{.64cm} - \sum_{j \in [h(n)]} c_{n, j}z_{n, j} - b(v_{n - 1}, v_n)\\
	&= g^*(v_{n - 3}) - \sum_{j \in [h(n - 2)]} c_{n - 2, j}z_{n - 2, j} - b(v_{n - 3}, v_{n - 2}) \push\hspace{.64cm} - \sum_{j \in [h(n - 1)]} c_{n - 1, j}z_{n - 1, j} - b(v_{n - 2}, v_{n - 1}) \push\hspace{.64cm} - \sum_{j \in [h(n)]} c_{n, j}z_{n, j} - b(v_{n - 1}, v_n)\\
	&= \dots\\
	&= g^*(v_0) - \sum_{i \in [n]} \left(b(v_{i - 1}, v_{i}) + \sum_{j \in [h(i)]} c_{i, j}z_{i, j}\right),
	\end{align*}
	so our inductive definition agrees with the given value of $g^*(v_n)$ by (\ref{equGStarBasis}). Finally, observe that, for any $i \in [n]$ and any arbitrary $g_{i - 1}, g_i, z \in \fm$,
	\begin{align*}
	&& x_{v_{i}}^{g_{i}} - x_{v_{i - 1}}^{g_{i - 1}} = z \\&&\txt{ is an equation in } \mathcal{G}(U_1)\\
	\iff&& (x_{v_{i}} + g_{i}) - (x_{v_{i - 1}} + g_{i - 1}) = z \\&&\txt{ is an equation in } U_1\\
	\iff&& (x_{v_{i}} + g_{i} + g^*(v_{i - 1})) - (x_{v_{i - 1}} + g_{i - 1} + g^*(v_{i - 1})) = z \\&&\txt{ is an equation in } U_1\\
	\iff&& (x_{v_{i}} + g_{i} + g^*(v_{i - 1})) - (x_{v_{i - 1}} + g_{i - 1} + g^*(v_{i - 1})) = z + \sum_{j \in [h(i)]} c_{i, j}z_{i, j} \\&&\txt{ is an equation in } U_1 \\&&\snc{\sum_{j \in [h(i)]} c_{i, j}z_{i, j} \in Z(v_{i - 1}, v_{i})}\\
	\iff&& (x_{v_{i}} + g_{i} + g^*(v_{i - 1})) - (x_{v_{i - 1}} + g_{i - 1} + g^*(v_{i - 1})) = z + \sum_{j \in [h(i)]} c_{i, j}z_{i, j} + b(v_{i - 1}, v_{i}) \\&&\txt{ is an equation in } U_2\\
	\iff&& (x_{v_{i}} + g_{i} + g^*(v_{i - 1}) - \sum_{j \in [h(i)]} c_{i, j}z_{i, j} - b(v_{i - 1}, v_i)) - (x_{v_{i - 1}} + g_{i - 1} + g^*(v_{i - 1})) = z \\&&\txt{ is an equation in } U_2\\
	\iff&& x_{v_{i}}^{g_{i} + g^*(v_{i - 1}) - \sum_{j \in [h(i)]} c_{i, j}z_{i, j} - b(v_{i - 1}, v_i)} - x_{v_{i - 1}}^{g_{i - 1} + g^*(v_{i - 1})} = z \\&&\txt{ is an equation in } \mathcal{G}(U_2)\\
	\iff&& x_{v_{i}}^{g_{i} + g^*(v_{i})} - x_{v_{i - 1}}^{g_{i - 1} + g^*(v_{i - 1})} = z \\&&\txt{ is an equation in } \mathcal{G}(U_2)\\
	\iff&& f(x_{v_{i}}^{g_{i}}) - f(x_{v_{i - 1}}^{g_{i - 1}}) = z \\&&\txt{ is an equation in } \mathcal{G}(U_2),
	\end{align*}
	so $f$ is a partial isomorphism over the entire path.
\end{proof}

We are finally ready to prove $C^k$-equivalence. Duplicator's strategy is to maintain consistency between $\mathcal{G}(U_1)$ and $\mathcal{G}(U_2)$ over a minimal tree in $H$ spanning all \emph{pebbled vertices} (that is, vertices $v \in V(H)$ such that some variable $x_v^g$ is pebbled in one of the two structures). Since $H$ has high girth relative to $r$, Spoiler will never be able to expose a cycle of inconsistency, as any cycle will have to contain a long path over which Duplicator can use Lemma~\ref{lemNUGLGRadiusProperty} to define a partial isomorphism regardless of any predetermined mappings at the endpoints.

\begin{lem}\label{lemNUGLGIndistinguishable}
	$\mathcal{G}(U_1) \equiv_{C^k} \mathcal{G}(U_2)$.
\end{lem}
\begin{proof}
	It is without loss of generality to assume $H$ is connected, for otherwise Duplicator can apply the strategy presented here on each connected component separately. On every round $i$ of the $k$-pebble bijective game played on $\mathcal{G}(U_1)$ and $\mathcal{G}(U_2)$, for any $u \in V(H)$, let $T_i(u)$ be a minimal tree containing $u$ and all pebbled vertices just after Spoiler has picked up a pebble. Let $P_i(u)$ denote the set of all of the vertices in $T_i(u)$ which have degree at least 3 in $T_i(u)$ or contain a pebbled vertex, also including $u$. Define $T_i := T_i(u^*_i)$ and $P_i := P_i(u^*_i)$, where $u^*_i$ is the new vertex pebbled in round $i$. Finally, define the forest $F_i(u)$ to be the subgraph of $T_i(u) \setminus T_{i - 1}$ (what this notation means is, remove all edges in $T_{i - 1}$ from $T_i(u)$, then remove isolated vertices) consisting of all segments in $T_i(u) \setminus T_{i - 1}$ between vertices in $P_i(u) \cup V(T_{i - 1})$ which have length less than $r$. See Figure~\ref{figBigTree} for an example.
	
	\begin{figure*}\begin{center}
			\includegraphics[scale=.73]{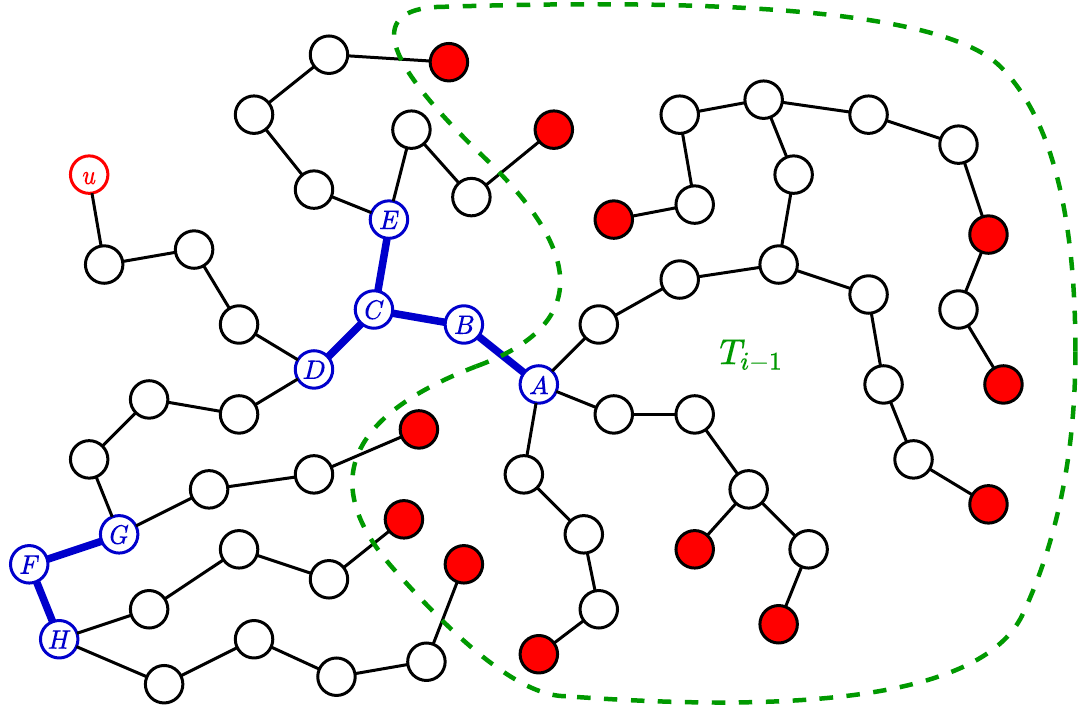}
			\caption{\label{figBigTree}The tree consisting of all vertices and edges in the figure is $T_i(u)$. This is a minimal tree that includes all pebbled vertices, which are filled in red, and vertex $u$, which is near the top left corner (also in red). The green dashed line outlines the boundary of $T_{i - 1}$ (not all vertices and edges of this tree are shown, just those that are also in $T_i(u)$). The set of vertices $P_i$ is not shown, but consists of all of the red vertices and vertices of degree $\geq 3$. Assuming that $r = 3$ (which is, of course, not nearly large enough; this is just for the purpose of illustration), the forest $F_i(u)$ is as depicted in blue, consisting of the lettered vertices $A$ through $H$ and all of the edges between those vertices.}
	\end{center}\end{figure*}

	A few remarks are in order about the conceptual meaning of $T_i(u)$ and $F_i(u)$. First, the parameter $u$ is included because Duplicator can essentially play an entirely different strategy for each possible vertex $u$ that Spoiler may choose to pebble. In other words, when determining the bijection among variables of the form $x_u^g$, Duplicator imagines what the new minimal tree $T_i$ between all pebbled vertices will look like if Spoiler chooses to put a pebble on one of those variables. We call this tree $T_i(u)$. Duplicator will imagine an entirely different tree for other vertices $u$, requiring an entirely different strategy. No matter which vertex $u^*$ is eventually chosen, Duplicator seeks to always maintain a partial isomorphism among all variables involved $T_i$. The trick is showing how to preserve this property as the tree changes from $T_{i - 1}$ in round $i - 1$ to $T_i$ in round $i$. Where the two trees overlap, we leave the bijection the same as in the previous round. Thus we only need to worry about the edges that are in $T_i$ but not in $T_{i - 1}$. We decompose this set of edges into a forest $F_i$ of short paths and junctions of degree $\geq 3$, and the remains, which must be a disjoint union of long paths (specifically, of length $\geq r$). Duplicator can use Lemma~\ref{lemNUGLGRadiusProperty} to define the bijection over these long paths regardless of the constraints on the endpoints (the constraints could come from vertices in $T_{i - 1}$, where Duplicator must not change the bijection from the previous round). Thus, $F_i$ is really the only part where Duplicator might run into trouble. Specifically, it could be impossible to extend a partial isomorphism over a connected component of $F_i$ if multiple endpoints are fixed from $T_{i - 1}$. The next two Claims rule out this possibility on the grounds that paths in $F_i$ are short and $H$ has high girth. Thus, each component of $F_i$ has at most one vertex with a constraint on the bijection. Since there are no cycles, it is possible to inductively extend the bijection to each connected component.
	
	\begin{clm}\label{lemTreePath}
		On any round $i$, for any vertex $u \in V(H)$, any path in $T_i(u)$ passes through at most $k$ vertices in $P_i(u)$.
	\end{clm}
	\begin{proof}[Proof of Claim]
		Let $p = (v_0, \seq{v}{n})$ be a path in $T_i(u)$. Consider the following map $h: p \cap P_i(u) \to V(T_i(u))$:
		\[h(v) := \threecases{\txt{if } \deg(v) < 3 \txt{ in } T_i(u)}{v}{}{\txt{a pebbled vertex (or $u$) reachable from $v$ by a path in}}{\txt{if } \deg(v) \geq 3 \txt{ in } T_i(u)}{\txt{$T_i(u)$ not passing through neighboring vertices in $p$}}\]
		Note that such a vertex in the second case above always exists when $v$ has degree at least 3, and is necessarily different from all other vertices in the image of $h$. Thus, $h$ is injective. Also, since vertices in $P_i(u)$ of degree less than 3 must be pebbled (or $u$), the output of $h(v)$ must always be a pebbled vertex (or $u$). Thus, we have an injection from $p \cap P_i(u)$ to a set of pebbled vertices (plus $u$), of which there are at most $k$ (since one pebble pair has been picked up), so $\abs{p \cap P_i(u)} \leq k$.
	\end{proof}
	
	\begin{clm}\label{lemGirth}
		On any round $i$, for any vertex $u \in V(H)$, there does not exist any path contained in $F_i(u)$ with both endpoints in $T_{i - 1}$.
	\end{clm}
	\begin{proof}[Proof of Claim]
		Suppose toward a contradiction that there was such a path $p_0$, joining $v_1, v_3 \in T_{i - 1}$. Let $v_2$ be the first vertex in $p_0$ that is contained within $T_{i - 1}$, excluding $v_1$ (it could just be $v_3$ if there are no earlier places where $p_0$ crosses $T_{i - 1}$). Let $p_1$ be the subpath of $p_0$ from $v_1$ to $v_2$. Since $T_{i - 1}$ is connected, there must be some path $p_2$ joining $v_1$ and $v_2$ in $T_{i - 1}$. Aside from $v_1$ and $v_2$, the path $p_1$ lies outside of $T_{i - 1}$ (since it is contained in $F_i$) whereas $p_2$ lies inside of $T_{i - 1}$, so $p_1$ and $p_2$ form a cycle. Since $p_1$ is contained within $T_i(u)$, by Claim~\ref{lemTreePath} it intersects at most $k$ vertices in $P_i(u)$, splitting $p_1$ up into at most $k + 1$ segments. As $p_1$ is contained in $F_i(u)$, the length of each of these segments is strictly less than $r$, so $p_1$ has length strictly less than $(k + 1)r$. Since $p_2$ is contained within $T_{i - 1}$, which is minimal, $p_2$ cannot contain any subpaths of length $(k + 1)r$ which do not intersect $P_{i - 1}$, for otherwise we could delete such a subpath and add $p_1$ to get a strictly smaller tree that still connects all vertices in $P_{i - 1}$, contradicting the minimality of $T_{i - 1}$. Applying Claim~\ref{lemTreePath} to round $i - 1$ and vertex $u^*_{i - 1}$, we have that at most $k$ vertices of $p_2$ intersect $P_{i - 1}$, so $p_2$ has length at most $k (k + 1) r$. Thus, concatenating $p_1$ and $p_2$ yields a cycle of size strictly less than
		\[(k + 1) r + k (k + 1) r = (k + 1)^2 r\]
		in $H$. This contradicts the fact that $H$ was constructed to have girth at least $(k + 1)^2 r$. Hence, no such path $p_1$ can exist.
	\end{proof}
	
	Returning to the proof of $C^k$-equivalence, let $X_i(u)$ denote the variable set of $\mathcal{G}(U_1)$ and $\mathcal{G}(U_2)$ restricted to $T_i(u)$, that is,
	\[X_i(u) := \{x_v^g \suchthat v \in T_i(u),\ g \in \fm\}.\]
	On each round $i$, Duplicator's will begin by defining a map
	\[g^*(i, u, \cdot): V(T_i(u)) \to \fm\]
	for each $u \in V(H)$. This map will determine how the bijection acts on variables involved in vertex $u$ via (\ref{equKEquals2PreserveProp}), which we will recall shortly. As we remarked near the beginning of the proof, this may be different for each possible $u$ that Spoiler could pebble. We will prove (inductively) that Duplicator can always define $g^*$ so that it has the following two properties:
	\begin{enumerate}
		\item\label{itmNUGLGProp1} For any pebbled vertex $v \in V(H)$, \[g^*(i, u, v) = g^*(i - 1, u^*_{i - 1}, v).\]
		\item\label{itmNUGLGProp2} The map $f_{i, u}: X_i(u) \to X_i(u)$ defined by
		\[f_{i, u}(x_v^g) := x_v^{g + g^*(i, u, v)}\]
		gives a partial isomorphism between $\mathcal{G}(U_1)$ and $\mathcal{G}(U_2)$ over pairs of variables whose underlying vertices are adjacent in $T_i(u)$.
	\end{enumerate}
	Obviously these properties hold at the beginning of the game, when there are no pebbled vertices and the tree is empty, so $g^*$ is the empty map. We will show that Duplicator can preserve these properties from one round to the next.

	On round $i$, Duplicator presents Spoiler with the bijection
	\[f_i(x_v^g) := x_v^{g + g^*(i, v, v)},\]
	which respects existing pebble pairs by property (\ref{itmNUGLGProp1}). No matter which vertex $u^*_i$ Spoiler chooses, the map $f_{i, u^*_i}$ agrees with $f_i$ over $u^*_i$, so we know that $f_{i, u^*_i}$ respects all pebble pairs since $f_i$ does. To see that Duplicator's bijection is a partial isomorphism between pebbled variables, first observe that, at the beginning of the next round $i + 1$, any edge between a pair of adjacent pebbled vertices of $H$ is a path of length one (which is less than $r$) between vertices in $T_i$. Thus, this edge must not be in $F_{i + 1}$, otherwise it would violate Claim~\ref{lemGirth} (applied to round $i + 1$). The only way that such an edge can fail to be in $F_{i + 1}$ is for it to have been in the previous tree, $T_i$. Thus, we have shown that, after Spoiler places the pebbles in round $i$, every pair of pebbled variables in the same relation involves vertices that are adjacent in $T_i$. Therefore, Property (\ref{itmNUGLGProp2}) guarantees Spoiler does not win on round $i$.
	
	All that remains is to show how Duplicator can satisfy properties (\ref{itmNUGLGProp1}) and (\ref{itmNUGLGProp2}) on each round $i$, assuming inductively that they are satisfied on round $i - 1$. Fix a vertex $u \in V(G)$. Duplicator defines $g^*(i, u, \cdot)$ in three steps: first over $V(T_i(u)) \cap V(T_{i - 1})$, then over $V(F_i(u)) \setminus V(T_{i - 1})$, then finally, over the remaining vertices $(V(T_i(u)) \setminus V(T_{i - 1})) \setminus V(F_i(u))$.
	
	Over $V(T_i(u)) \cap V(T_{i - 1})$, Duplicator simply sets
	\[g^*(i, u, v) := g^*(i - 1, u^*_{i - 1}, v),\]
	which is well-defined over $V(T_{i - 1})$ and clearly satisfies both properties (\ref{itmNUGLGProp1}) and (\ref{itmNUGLGProp2}), inductively assuming that $g^*(i - 1, u^*_{i - 1}, \cdot)$ did. Since $V(T_i(u)) \cap V(T_{i - 1})$ contains all pebbled vertices, we no longer have to worry about property (\ref{itmNUGLGProp1}); we just have to define $g^*(i, u, \cdot)$ on the remainder of $V(T_i(u))$ so that property (\ref{itmNUGLGProp2}) is satisfied.
	
	Duplicator then uses the following algorithm to define $g^*(i, u, \cdot)$ over $V(F_i(u)) \setminus V(T_{i - 1})$. Basically, we define $g^*$ by inductively extending from neighboring vertices. A vertex is assigned the value it needs to have in order to make Duplicator's bijection a partial isomorphism across each edge in $F_i$.
	
	\begin{algorithm}[H]
		\While{\textbf{true}}
		{
			\uIf{there exists $\{v_1, v_2\} \in E(F_i(u))$ such that $g^*(i, u, v_1)$ is defined but $g^*(i, u, v_2)$ is not defined}
			{
				$g^*(i, u, v_2) \gets g^*(i, u, v_1) + b(v_1, v_2)$\label{linDefineGStarAlgo}\;
			}
			\uElseIf{there exists $v \in V(F_i(u))$ such that $g^*(i, u, v)$ is not defined}
			{
				$g^*(i, u, v) \gets$ anything\;
			}
			\Else
			{
				\KwRet{}\;
			}
		}
	\end{algorithm}
	
	Observe that the constraints involving each edge in $F_i(u)$ considered in the first case are preserved by $f_{i, u}$: for all $g_1, g_2, z \in \fm$,
	\begin{align*}
	&& x_{v_1}^{g_1} - x_{v_2}^{g_2} = z \\&&\txt{ is an equation in } \mathcal{G}(U_1)\\
	\iff&& (x_{v_1} + g_1) - (x_{v_2} + g_2) = z \\&&\txt{ is an equation in } U_1\\
	\iff&& (x_{v_1} + g_1) - (x_{v_2} + g_2) = z + b(v_1, v_2) \\&&\txt{ is an equation in } U_2\\
	\iff&& (x_{v_1} + g_1 + g^*(i, u, v_1)) - (x_{v_2} + g_2 + g^*(i, u, v_1) + b(v_1, v_2)) = z \\&&\txt{ is an equation in } U_2\\
	\iff&& x_{v_1}^{g_1 + g^*(i, u, v_1)} - x_{v_2}^{g_2 + g^*(i, u, v_1) + b(v_1, v_2)} = z \\&&\txt{ is an equation in } \mathcal{G}(U_2)\\
	\iff&& x_{v_1}^{g_1 + g^*(i, u, v_1)} - x_{v_2}^{g_2 + g^*(i, u, v_2)} = z \\&&\txt{ is an equation in } \mathcal{G}(U_2)\\
	&&\txt{(from line~\ref{linDefineGStarAlgo} of the algorithm)}\\
	\iff&& f_{i, u}(x_{v_1}^{g_1}) - f_{i, u}(x_{v_2}^{g_2}) = z \\&&\txt{ is an equation in } \mathcal{G}(U_2).
	\end{align*}
	
	For example, if $F_i(u)$ is as in Figure~\ref{figBigTree}, then the first iteration of the algorithm would define $g^*(i, u, B)$ so that the constraints involving $A$ and $B$ are consistent under $f_{i, u}$. The next iteration would then define $g^*(i, u, C)$ so that the constraints involving $B$ and $C$ are consistent. Similarly, the next two iterations would set $g^*(i, u, D)$ and $g^*(i, u, E)$ (these could happen in either order). On the fifth iteration, we would hit the second case of the algorithm and set one of $g^*(i, u, F)$, $g^*(i, u, G)$ or $g^*(i, u, H)$ arbitrarily. The final two iterations would set the other two values according to the first case.
	
	Since the edges encountered in the first case are always made consistent, the only way that $f_{i, u}$ could fail to be a partial isomorphism over $F_i(u)$ is if, at some iteration, there were two different edges satisfying the condition in the first case that yielded different values for $g^*(i, u, v_2)$ for some $v_2$. Since $F_i(u)$ is a forest, the only way that this could happen is if some connected component of $F_i(u)$ had two distinct vertices $v_1$ and $v_2$ on which $g^*(i, u, \cdot)$ was already defined before the algorithm started, which can only happen if $v_1, v_2 \in V(T_{i - 1})$. But this means that there is a path in $F_i(u)$ from $v_1$ to $v_2$ that violates Claim~\ref{lemGirth}. Thus, property (\ref{itmNUGLGProp2}) is still satisfied.
	
	At this point, the only remaining edges of $T_i(u)$ which Duplicator needs to worry about are those which are in $T_i(u) \setminus T_{i - 1}$ but are not in $F_i(u)$. By the definition of $F_i(u)$, this consists of paths of length at least $r$, each with a disjoint set of intermediate vertices. Since $g^*(i, u, \cdot)$ has not yet been defined on any of the intermediate vertices, Duplicator can apply Lemma~\ref{lemNUGLGRadiusProperty} to each one separately. Thus, property (\ref{itmNUGLGProp2}) is satisfied over the entirety of $T_i(u)$.
	
	We have shown that Duplicator has a winning strategy in the $k$-pebble bijective game played on $\mathcal{G}(U_1)$ and $\mathcal{G}(U_2)$, so, by Theorem~\ref{thmHella}, $\mathcal{G}(U_1) \equiv_{C^k} \mathcal{G}(U_2)$.
\end{proof}

\subsection{Proof of main result}\label{subMainResultAndCorollary}

\begin{thm}\label{thmMain}
	For any $\delta > 0$ and any positive integer $\ell$, there exists a positive integer $q$ such that no sentence of FPC distinguishes \UG($q$) instances (encoded as $\tauuug$-structures) of optimal value $\geq \frac{1}{2^\ell}$ from those of optimal value $< \soundness + \delta$.
\end{thm}

\begin{proof}
	Let $\delta$ and $\ell$ be given, pick an arbitrary $\varepsilon \in (0, \frac12)$, then let $q$ be as defined in Section~\ref{subConstruction}. Suppose toward a contradiction that there is an FPC sentence $\phi$ distinguishing the two cases. Then $\phi$ can be translated into a logically equivalent sentence $\phi'$ of $C^k$ for some fixed $k$. Given this value of $k$, let $\mathbb{A}_k := \mathcal{G}(U_1)$ and $\mathbb{B}_k := \mathcal{G}(U_2)$ be as defined in Section~\ref{subConstruction}. By Lemmas~\ref{lemNUGLGCompleteness},~\ref{lemNUGLGSoundnessU2} and~\ref{lemNUGLGIndistinguishable}, with probability at least $1 - 2\varepsilon$, we will have that \begin{align*}
	\opt(\mathbb{A}_k) \geq \frac{1}{2^\ell}, && \opt(\mathbb{B}_k) < \soundness + \delta,
	\end{align*}
	yet $\mathbb{A}_k \equiv_{C^k} \mathbb{B}_k$. Since $1 - 2\varepsilon > 0$, this implies that there is \emph{some} pair of $\tauuug$-structures $(\mathbb{A}_k, \mathbb{B}_k)$ produced by this construction satisfying these properties. This is a contradiction, since $\phi'$ cannot distinguish $\mathbb{A}_k$ and $\mathbb{B}_k$ as they are $C^k$-equivalent, so $\phi$ cannot distinguish them either.
\end{proof}

\begin{cor}\label{corUGLowGapMain}
	For any $\alpha \in (0, 1]$, for sufficiently large $q$ there does not exist an FPC-definable $\alpha$-approximation algorithm for \UG$(q)$.
\end{cor}
\begin{proof}
	Let $\alpha \in (0, 1]$ be given. Define
	\begin{align*}
	\ell := \left\lceil 2 - \log_2(\alpha) \right\rceil, &&
	\delta := \soundness.
	\end{align*}
	Note that $\ell$ is a positive integer and $\delta > 0$. Let
	\begin{align*}
	c &:= \frac{1}{2^\ell},\\
	s &:= \soundness + \delta = \frac{2}{2^{2\ell - 1} + 2^{\ell - 1}}
	\end{align*}
	Then, by Theorem~\ref{thmMain}, there exists a positive integer $q$ such that no sentence of FPC distinguishes $\tauuug$-structures with optimal value $\geq c$ from those with optimal value $< s$. Observe that 
	\begin{align*}
	\frac{s}{c} &= \frac{\frac{2}{2^{2\ell - 1} + 2^{\ell - 1}}}{\frac{1}{2^\ell}}\\
	&= \frac{2^{\ell + 1}}{2^{2\ell - 1} + 2^{\ell - 1}}\\
	&\leq \frac{2^{\ell + 1}}{2^{2\ell - 1}}\\
	&= 2^{2 - \ell}\\
	&\leq 2^{2 - \left(2 - \log_2(\alpha)\right)} \snc{\ell \geq 2 - \log_2(\alpha)}\\
	&= \alpha.
	\end{align*}
	Therefore, it follows from Lemma~\ref{lemGapToInapproximability} that there is no FPC-definable $\alpha$-approximation algorithm for \UG($q$).
\end{proof}$ $

\section{Conclusion}\label{secConclusion}

This paper has two objectives. The first is to introduce a new proof idea into the arsenal of tools for CFI-constructions. Our main challenge was that, to prove an FPC-inapproximability result, it is not sufficient to have Duplicator maintain the standard kind of invariant used in CFI-constructions: that the two structures are always isomorphic by a bijection agreeing with all pebble pairs, except at one place where the isomorphism breaks down. Using techniques from graph theory and linear algebra, we have shown how Duplicator can win the $k$-pebble bijective game while only maintaining an extremely weak invariant---consistency over a very small portion of the two structures. Hopefully this approach can be extended to other approximation settings as well.

The second objective is to make progress on the FPC-version of the Unique Games Conjecture. We have established an FPC-inapproximability gap which is competitive with recent results in ordinary complexity theory about polynomial-time approximation algorithms, yet our result holds without the assumption that $\PP \neq \NP$. Of course, our proof method is entirely different than the proof methods used in ordinary complexity theory---so different that the fundamental problem which is shown to be inexpressible in FPC, distinguishing $\mathcal{G}(U_1)$ from $\mathcal{G}(U_2)$, is not even $\NP$-hard.\footnote{To see this, observe that a given bundle of constraints in either structure is satisfiable if and only if a certain system of $m - \ell$ linear equations over $\fm$ is solvable. (Specifically, the equations express the constraint that $x^{g_1}_{v_1} - x^{g_2}_{v_2}$ in the $\ell$-dimensional affine subspace $Z(v_1, v_2) - b(v_1, v_2) + (g_1 - g_2)$.) In $\mathcal{G}(U_1)$, the union of all of these systems is completely satisfiable, while in $\mathcal{G}(U_2)$, they are not, so distinguishing $\mathcal{G}(U_1)$ from $\mathcal{G}(U_2)$ can be accomplished by Gaussian elimination.}

To prove the FPC-UGC in its entirety, it will be necessary to (almost completely) eliminate the parallel, contradictory constraints between pairs of vertices. These simple gadgets provide Duplicator with the power to make choices, circumventing the ``uniqueness" inherent in ``Unique Games," but they kill the completeness of the inapproximability gap: the more-satisfiable instance will still only have satisfiability at most $\frac{1}{2}$. One possible approach would be to build larger gadgets with more complicated automorphisms achieving a similar effect as these parallel constraints. We leave this possibility to future work.

\section*{Acknowledgments}
\noindent This work was generously supported by the Winston Churchill Foundation through a Churchill Scholarship at the University of Cambridge, UK. I would like to express my deep gratitude to Anuj Dawar, whose advice and insights substantially improved the results in this paper. I am also very grateful to Neil Immerman, as well as my LMCS and LICS reviewers, for detailed, helpful comments on earlier drafts.

\bibliographystyle{alphaurl}
\bibliography{bibliography}

\end{document}